\documentclass[11pt,reqno]{amsart}

\usepackage{amsmath, amsthm, amsopn, amssymb, graphicx, enumerate, geometry, caption, subcaption, bbm, color, textcomp, afterpage, cancel}

\usepackage{centernot, mathtools}

\newtheorem{thm}{Theorem}[section]
\newtheorem{lemma}[thm]{Lemma}
\newtheorem{prop}[thm]{Proposition}
\newtheorem{definition}[thm]{Definition}

\theoremstyle{definition}

\numberwithin{equation}{section}

\def\={\  =\  }
\def\<{ \  \leq \  }
\def\0{\textbf{0}}

\newcommand{\R}{\mathbb{R}}
\newcommand{\Z}{\mathbb{Z}}

\newcommand{\E}{\mathbb{E}}

\newcommand{\1}{\mathbbm{1}}

\def\eps{\varepsilon}
\renewcommand\P{\mathbb{P}}

\newcommand{\extB}{\partial_{\mathrm{v}}}

\renewcommand{\vert}{\text{V}}
\newcommand{\edge}{\text{E}}

\newcommand{\e}{\cancel{e}}

\DeclareMathOperator\zero{\mathbf{0}}

\DeclareMathOperator{\Leb}{Leb}

\def\={\  = \  }
\def\be{\begin{equation}}
\def\ee{\end{equation}}
\def\ba{\be \begin{aligned}}
\def\ea{\end{aligned} \ee}

\def\G{\mathbb{G}}
\def\L{\mathcal{L}}
\def\zero{\textbf{0}}

\date{\normalsize Version of \today}

\begin{document}

\author{Michael Aizenman}
\address{Michael Aizenman\hfill\break
    Departments of Mathematics and Physics,
    Princeton University,
    Princeton, NJ 08544,
    USA.
\footnote{Weston Visiting Professor, Weizmann Institute of Science, Rehovot Israel.}}
\email{aizenman@princeton.edu}

\author{Matan Harel}
\address{Matan Harel\hfill\break
    School of Mathematical Sciences,
    Tel Aviv University,
    Israel.}
\email{matanharel8@tauex.tau.ac.il}

\author{Ron Peled}
\address{Ron Peled\hfill\break
    School of Mathematical Sciences,
    Tel Aviv University,
    Israel.}
\email{peledron@tauex.tau.ac.il}

\title[Exponential decay of correlations in the $2D$ RFIM]
{Exponential decay of correlations \\ in the $2D$ random field Ising model}

\date{\today}

\begin{abstract}
An extension of the Ising spin configurations to continuous functions is used for an exact
representation of  the Random Field Ising Model's order parameter in terms of disagreement percolation.   This  facilitates an extension of the recent analyses of the decay of correlations to positive temperatures, at homogeneous but arbitrarily weak disorder.
 \end{abstract}
\maketitle

\section{Introduction}

\subsection{The Imry-Ma phenomenon and the question discussed here}
\mbox{} \\[-1ex]

In the Imry-Ma phenomenon, first-order phase transitions of two-dimensional statistical mechanics systems are turned  continuous under the incorporation of quenched disorder in the suitably conjugate field variable~\cite{IM75}. This is exemplified by the rounding of the Ising model's famed discontinuity of magnetization as a function of the external field $h$, through the addition of a random field of arbitrarily weak strength $\eps >0$. \\

The Random-Field Ising Model (RFIM) is a spin system on a finite graph with the  Hamiltonian
\be \label{H}
  H(\sigma) := - \sum_{\substack{\{u,v\}\in\edge(\G) }}
J_{u,v}\,  \sigma_u \sigma_v  - \sum_{v\in \vert(\G)} (h+ \eps\,\eta_v) \sigma_v \,.
\ee
where $\{\sigma_v\}_{v\in \vert(\G)} $ are the Ising $\pm 1$-valued variables, $\{J_{u,v}\}$ are positive coupling constants, $\{\eta_v\}_{v\in \vert(\G)}$ are independent, identically distributed, non-constant external fields, $\edge(\G)$ is the graph's edge set, and $\eps$ is the disorder strength parameter.

For a fixed realization of the random field, $\eta$, the Gibbs equilibrium state is described by the probability measure
\be\label{eq:MeasureDef}
\mathbb{P}_{\G}(\sigma) := \frac{1}{Z_{\G}} \exp\left(-\beta H(\sigma)\right),
\ee
where $\beta  = 1/T \in [0, \infty)$ is the inverse temperature and $Z_{\G}$ is  the normalizing factor, which is referred to as the partition function.   In finite volumes the Gibbs measure depends also on the boundary conditions, e.g. the imposition of specific  values on the spins along the boundary.   Given $A \subset \G$, and a function $\tau: A \rightarrow \{-1,+1\}$, we denote by $\mathbb{P}_{\G}^{A, \tau}$ the Gibbs measure on $\G$ restricted to configurations such that $\sigma = \tau$ on $A$; its normalizing factor is
\be
Z_{\G}^{A,\tau} := \sum_{\sigma \in \{-1,+1\}^{\vert(\G)}} \exp\left(- \beta H(\sigma)\right) \cdot \1[\sigma = \tau \text{ on A}].
\ee
If $\G$ is a subset of $\mathbb{Z}^2$, $A$ is its internal vertex boundary (denoted $\extB \G$), and $\tau$ is identically $+1$, we will abbreviate $\mathbb{P}_{\G}^{A, \tau}$ by $\mathbb{P}_{\G}^{+}$; the negative case is defined analogously. The associated expectation operators will be denoted by $\langle \cdot \rangle_{\G}^{+}$ and $\langle \cdot \rangle_{\G}^{,-}$.

For the RFIM, the $+$ and $-$ boundary conditions are of particular interest, since, for any $A$, the Gibbs states corresponding to all other values of $\tau$ are bracketed by these two (through the FKG inequality).    In this case, the formal expression of the rounding effect is the statement that, in two dimensions and for arbitrarily weak disorder strength $\eps\neq 0$, there is no residual symmetry breaking, in the sense that
\be
\lim_{\Lambda \to \Z^2}  \mathbb{P}_{\Lambda}^{ +} \ = \  \lim_{\Lambda \to \Z^2}  \mathbb{P}_{\Lambda}^{ -}
\ee
as probability measures (or, equivalently, as expectation value functionals) describing the system in the infinite volume limit.

For the Ising model, equality of the limiting $\pm$ states is equivalent to the uniqueness of the infinite volume Gibbs state.
To place this in a broader context, let us add that, for systems which are not endowed with an FKG-type monotonicity property, one has only the weaker statement that the free energy function is differentiable in $h$, or, equivalently, that all Gibbs states assign a common value to the volume-averaged magnetization~\cite{AW89, AW90}.  However, the  robustness of the thermodynamic version of the statement has enabled its extension to quantum systems~\cite{AGL12}.

For spin models with rotational symmetry, such as the $O(N)$ models, the rounding effect (at any $\eps>0$) extends  to three and four dimensions, provided the probability distribution of $\eta \in \R^N$ is rotationally invariant \cite{AW89,AW90}. Furthermore, all the statements also hold in the limit $\beta \to \infty$, in which case they refer to the model's infinite-volume ground state(s).  \\

An intuitive explanation of the effect's dependence on dimension can be found by comparing the contribution of two possibly conflicting terms to the difference of free energies between two ordered states in a box of linear size $L$.  One term is the  coupling with the random field throughout the volume of the box, and  the other is the  coupling with the neighboring regions.  The first produces a local magnetization bias, while the second tends to keep the local magnetizations aligned. It seems natural to guess that the effect of the local field is a Gaussian variable of scale $\eps L^{d/2}$, while the boundary pull is of the order of $L^{d-1}$, or, in the presence of rotational symmetry, $L^{d-2}$.
At the critical dimension, which is $d=2$ in case of the RFIM~\cite{Imb85,BK87,BK882}, the two scale similarly.  However, for any given site $u\in \Z^2$, fluctuations will cause the random field in a box centered at $u$ to eventually exceed the boundary pull.  This suggests that, even at weak disorder ($\eps \ll J$), the Gibbs state resembles a patchwork of locally ordered states. \\

Since the early works on the subject, questions have been raised concerning the decay rate of the resulting state's correlation functions.  These relate to the decay of correlations among the local spin magnetizations, which are quasi-local functions of the quenched random field, as well as to the typical decay of spin correlations within the quenched state.  As presented in~\cite{AP18}, both are dominated by the function
\begin{equation} \label{eq:m}
  \begin{split}
  m(L) \equiv         m(L;T, \mathcal J, h,  \eps)    \,  := \, \frac 1 2 \left[
    \E[\langle\sigma_\0\rangle_{\Lambda(L)}^{+} ] \ - \      \E[\langle\sigma_\0\rangle_{\Lambda(L)}^{ -}]
    \right] \,
  \end{split}
\end{equation}
where
\begin{equation}
  \Lambda_u(L):=\{v\in\Z^2\,\colon\, d(u,v)\le L\}\,\quad \mbox{,}\quad \  \Lambda(L) :=   \Lambda_\zero(L) \,,
\end{equation}
with $d(u,v)$  the graph distance on  $\Z^2$ and $\zero:=(0,0)$.

 It has been early recognized, and since then proven in a variety of ways,  that at \emph{high disorder} and/or high temperatures, $  m(L) $ decays exponentially fast \cite{ImbFro85,Ber85,CJN18}.  The more interesting question is the nature of the decay at weak disorder, and in particular, whether there is a qualitative difference between \emph{weak} and \emph{very weak} disorder, with transition from exponential decay at weak disorder to power-law decay at very weak disorder \cite{GM82,DS84,BK88}.  It is convenient and natural to focus on the case of Gaussian disorder.  This special case is fully expected to be indicative of more general distributions.

Recent results on the subject include:  i) an initial bound $m(L) \leq C(\eps)/\sqrt{\log \log L}$ \cite{C17},  ii) power-law bounds for all $\eps >0$ and  $T\geq0$, including also all finite-range ferromagnetic interactions \cite{AP18},
and, more recently,  iii) exponentially-decaying upper bounds for all  $\eps>0$~\cite{DX19}, limited  to the nearest-neighbor case and $T=0$.   The latter result answers (in the negative)  the question of possible transition in $\eps$  for the RFIM's ground state, but stops short of addressing the additional effect of thermal fluctuations.

The main result of this paper, which is being produced in parallel to an ongoing extension of~\cite{DX19} by its authors to positive temperatures, is the proof of the following statement.

\begin{thm}\label{thm:exponential_bound}
In the random-field Ising model on $\mathbb{Z}^2$ with the standard nearest-neighbor interaction of strength $J$ and
independent standard Gaussian random field  $(\eta_v)$,  for any temperature $T \geq 0$, uniform external field $h \in \mathbb{R}$, and disorder intensity $\eps>0$,
  \begin{equation} \label{eq:exponential_bound}
m(L; T , J, h, \eps) \le C(J/\eps) \cdot e^{-c(J/\eps) L} \, ,
  \end{equation}
with  $c,C >0$ that do not depend on $T, h$ or $L$.
\end{thm}
\medskip
\subsection{Key ingredients}

Following is an informal summary of the key concepts used in the proof.  These are presented more explicitly in the sections which follow. \\

\noindent \emph{Continuous extension of the Ising Model.}
We introduce an extension of the model in which the original binary spin variables are extended to random continuous functions supported on the graph's edges. The extended functions are restricted to $\pm1$ values only at the graph's vertices, and their measure's restriction to this collection of variables yields the Ising model's standard Gibbs state.   \\

\noindent \emph{Disagreement percolation.}
The term has already been invoked as offering a clarifying perspective on questions of the type considered here.
In this paper, we will invariably refer to a version which is based on the extended formulation of the model. The continuity of the extended spin function allows to identify the order parameter $m(L)$ as the suitably averaged probability that, in $\Lambda(L)$, the origin is connected to the boundary by a path along edges on which $\sigma^+_u >  \sigma^-_u$, with the pair $(\sigma^+, \sigma^-)$ sampled independently with the indicated boundary conditions.
 \\

 \noindent \emph{Surface tension.}
 We extend the past observations of~\cite{AP18} on the different relations of the surface tension with disagreement percolation. In particular, we present an upper bound on the surface tension between two surfaces in terms of the amount of the disagreement percolation flux through an arbitrary non-anticipatory random set separating the two.
 \\

 \noindent \emph{Tortuosity of the disagreement paths.}
 The improvement presented in this paper over \cite{AP18} is based on the quantification of the observation, which was expressed in \cite{AP18} at only a heuristic level, that the disagreement percolation is weakened by the fractality of its connected clusters.
A path from this suggestion to a proof was pointed out, in a somewhat modified context, in the aforementioned work of Ding-Xia~\cite{DX19}.  As there, this part of the argument makes an essential use of the sufficient condition for tortuosity of Aizenman-Burchard~\cite{AB99}.
Verifying the condition's assumptions is a somewhat technical part of the analysis, and the only step for which  the arguments employed here are limited to the nearest-neighbor models, i.e. restricted to strictly planar interactions, and not just short-range $2D$ models.

Through the tortuosity condition, we establish that, for some $\alpha>0$, the probability that a regular annulus of scale $\ell$ is traversed by a path of disagreement percolation with less  than $\ell^{1+\alpha}$ lattice steps vanishes as $\ell$ tends to infinity.
 \\

 \noindent \emph{A bootstrap argument.}   The last step in the proof is a bootstrap argument which extracts an exponential upper bound on the function $m(L)$ from the above estimates. Use is made of the fact that, like in the case of other percolation models, fast enough power law decay implies exponential decay (cf.~\cite{AP18}).  Anti-concentration bounds play a role in the analysis, as well as concentration bounds which control the contribution of exceptional random field configurations. \\

The above outline offers also the plan of the paper.

\section{A continuous extension of the Ising model}

While spins of a general Ising model are naturally associated with the vertices of a graph, which will generically be denoted $\G$, we find that, for a faithful disagreement percolation picture, it is convenient to extend their definition to functions supported along the graph's edges. For ease of notation, this construction is presented for the standard nearest-neighbor interaction. We extend $\G$ into a metric graph, i.e. substitute metric intervals for the edges. The resulting metric space of lines linked at the vertices of $\G$ will be denoted  $\mathcal L_\G$. We shall use $\e$ to denote the midpoint of the edge $e$.

The construction given below extends the notion of spin configurations on $\G$ to a set of highly-constrained  continuous functions on $\mathcal L_\G$, each being completely determined by its values on the vertices and the edge midpoints of   $\G$, where the function is allowed to take values only in $\{-1,0,+1\}$.\\

Formally, we let $ \Omega_\G$  be the set of continuous functions $\bar{\sigma}$ such that
\begin{enumerate}[i)]
\item  at the graph's vertices $\bar{\sigma} \in \{-1,+1\}$,
\item at edge midpoints, $\bar{\sigma} \in \{-1,0,+1\}$, and
\item along each edge, $\bar{\sigma}$ is given by the linear interpolation between its value at the center and at the nearest of the two edge ends.
\end{enumerate}

Thus, $\bar{\sigma}$ is determined by its restriction to the vertices of $\G$ and to the midpoints of its edges, which we denote by $\{\sigma_v\}_{v \in \vert(\G)}$ and $\{\kappa_{\e}\}_{e \in \edge(\G)}$, respectively. Let $\bar{\G}$ be the graph whose vertex set is the disjoint union of $\vert(\G)$  and the midpoints of the edges of $\G$. The edge set of $\bar{\G}$ is $\{\{v,\e\} : v \in e\}$. We refer to it as the extended graph.

On a finite graph $\G$, given a realization $\eta$ of the external field, let $\bar{\P}_{\bar{\G}}$ be the unique probability measure on the set of functions $ \Omega_\G$ with
\begin{equation}\label{eq:P_Lambda_tau_def}
  \bar{\P}_{\bar{\G}}(\bar{\sigma}) := \frac{1}{\bar{Z}_{\bar{\G}}} \prod_{ \substack{v\in \vert(\G) \\ e \in \edge(\G) \\ v \in e}}  W(\sigma_v,\kappa_{\e}) \cdot  e^{-\beta U(\sigma)}
%
\end{equation}
where
\be
 U(\sigma) : = \sum_{v \in \vert(\G)} (h + \eps \eta_v) \sigma_v,
\ee
and, for each $a \in\{-1,1\}$ and $ b\in\{-1,0,1\}$,
\begin{equation}\label{eq:Wdef}
    W(a,b):= \lambda \Big( \delta _{a,b}  \ + \ t \delta_{b,0}  \Big)\ = \  \lambda  \cdot \begin{cases}
    1& b=a,\\
    t& b =0,\\
    0&b =-a
  \end{cases}
\end{equation}
and $(t, \lambda)$ are related to the inverse temperature by
\begin{equation} \label{t_beta}
  t:=(e^{2 J \beta}-1)^{-\frac{1}{2}} \, , \qquad \lambda =  2 \sinh(J \beta)^{1/2} \,.
\end{equation}
One may note that if $e =\{u,v\}$ and $\sigma_u \neq \sigma_v$, the value of $\kappa_{\e}$ is constrained to be zero; furthermore, $\kappa_{\e} = +1$ implies that $\sigma_u = \sigma_v = +1$, with an analogous statement holding for~$-1$. We call these conditions the `hard' constraints of the extended Ising model.

    \begin{figure}
\begin{center}
\includegraphics[width=0.5\textwidth]{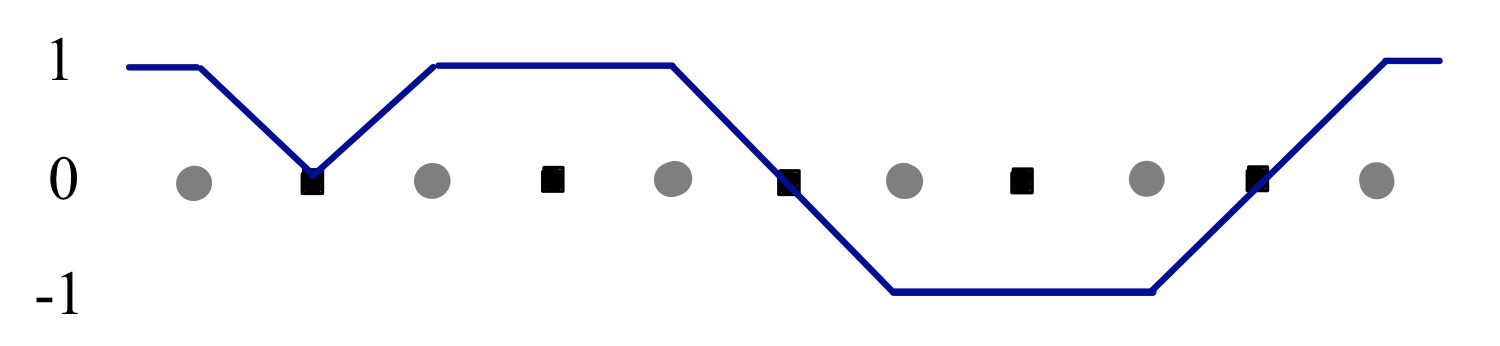}
\caption{A  possible configuration of the extended spin function $\bar \sigma$.  Dots represent  lattice sites, on which
$\bar\sigma_v$ takes only the values  $\pm1$.  Squares represent edge midpoints, on which $\bar \sigma_{\e}$ takes values in $\{-1,0,1\}$.}
\label{fig:extended}
\end{center}
\end{figure}

Given a subset $A \subset \bar{\G}$, we call $\tau: A \rightarrow \{-1,0,1\}$ an allowed configuration if there exists $\bar{\sigma}$ which is equal to $\tau$ on $A$ and satisfies the hard constraints described above.  For any $A$ and allowed configuration $\tau$, we denote the corresponding restricted probability measure $\bar{\P}_{\bar{\G}}^{ A , \tau}$. When $\beta < \infty$, this measure is equivalent to conditioning $\bar{\P}_{\bar{\G}}$ on the event $\{ \sigma = \tau \text{ on A} \}$. The partition function $\bar{Z}_{\bar{\G}}^{A,\tau}$ is defined analogously.

\begin{lemma}\label{lem:DMPExt} For any finite graph $\G$ and a realization of the external field $\eta$, at positive temperatures:
\begin{enumerate}[i)]
\item The probability measure $ \bar{\P}_{\bar{\G}}$ has the domain Markov property on the extended graph $\bar \G$.
\item The measure's restriction to $\vert(\G)$ coincides with the Gibbs measure $\mathbb{P}^{\G}$ on $\G$.
\item
For any $A_v \subset \vert(\G)$ and $\tau: A_v \rightarrow \{-1,+1\}$, the marginal distribution of the restricted measure $\bar{\P}_{\bar{\G}}^{ A_v , \tau}$ on $\vert(\G)$ coincides with the correspondingly restricted measure $\mathbb{P}_{\G}^{A_v,\tau}$.
\item
For any $A_e \subset \vert(\G)$, and allowed configuration $\tau$ such that $\tau |_{A_e} = 0$, the marginal distribution of $\bar{\P}_{\bar{\G}}^{A , \tau}$ on $\vert(\G)$ coincides with the corresponding Gibbs measure on the reduced graph obtained by limiting the edge set to $\edge(\G)\backslash A_e$ (without changing the vertex set).
\item The extension does not affect the partition function, i.e.
\be \label{Z_Zbar}
Z_{\G} = \bar{Z}_{\bar{\G}} \quad \text{ and } \quad Z_{\G}^{A,\tau} = \bar{Z}_{\bar{\G}}^{A,\tau} \text{ for any } A \subset \vert(\G).
\ee
\end{enumerate}
\end{lemma}

\begin{proof} The statement can be readily deduced from the definition, paying attention to the following key observations, listed by the claim number.\\
i) The domain Markov property follows from the multiplicative structure of $ \bar{\P}_{\bar{\G}}$. \\
ii) For each edge of $\G$, $e=\{u,v\} $, the summation over the  values of $\kappa_{\e}$ yields
\be
\sum _{{\kappa_{\e}}} \left( \delta _{\sigma_u,{\kappa_{\e}} } + t \delta_{\kappa_{\e},0}   \right)  \left(  \delta _{\sigma_v,{\kappa_{\e}} } + t \delta_{{\kappa_{\e}},0 }  \right)
= \delta _{\sigma_v,\sigma_u }  + t^2.
\ee
A consultation with \eqref{t_beta}
implies that
\be\label{eq:weightEquality}
\sum _{{\kappa_{\e}}} W(\sigma_u,\kappa_{\e}) W(\sigma_v,\kappa_{\e}) = e^{\beta J \sigma_u \sigma_v}.
\ee
iii) Since all hard constraints involve at least one mid-edge variable, restricting the value of $\sigma$ on $A \subset \vert(\G)$ cannot violate any of them. \\
iv) For each edge $\{u,v\}\in \edge(\G)$ the  factor $[W(\sigma_u,0) W(\sigma_v,0)]$ in \eqref{eq:P_Lambda_tau_def},  is spin independent  and thus can be omitted.   \\
v) The  last claim (and the only one for which the value of $\lambda$ in \eqref{eq:Wdef} is of relevance) follows readily from \eqref{eq:weightEquality}.
\end{proof}

As a corollary, we note that, in the limit $\beta \to \infty$, the probability distribution of $\{\sigma_v\}_{v\in \vert (\G)}$ concentrates on the ground-state configurations of $H(\sigma)$, and at edge midpoints the values are rigidly constrained to equal the mean of the value of the two endpoints.

Another corollary is that the extended Ising model inherits the standard FKG-type monotonocity properties of the (non-extended) Ising model. For instance, if $\Lambda_1 \subset \Lambda_2 \subset \G$, the marginal of the measure $\bar{\mathbb{P}}_{\bar{\Lambda}_2}^{\extB\Lambda_2, +}$ on $\bar{\Lambda}_1$ is stochastically dominated by $\bar{\mathbb{P}}_{\bar{\Lambda}_1}^{\extB\Lambda_1, +}$.

\section{Disagreement percolation}
The  extension of the Ising model to $\bar\G$ facilitates the following   exact relations in which the mean spin - spin correlation, and  related measures of influence propagation, are cast in terms of disagreement percolation.

With that end in mind, let $\G$ be a finite graph, $A \subset \vert(\G)$, and consider $\bar{\sigma}^+$ and $\bar{\sigma}^-$ to be sampled independently from $\bar{\P}_{\bar{G}}^{A,+}$ and $\bar{\P}_{\bar{G}}^{A,-}$, respectively. The {\em disagreement set} of $\bar{\sigma}^+$ and $\bar{\sigma}^-$ is
\be
\mathcal{D}:=\mathcal{D}(\bar{\sigma}^+,\bar{\sigma}^-) = \{u \in \bar{\G} : \bar{\sigma}_u^+ \neq \bar{\sigma}_u^- \}.
\ee

We say that $A,B \subset \bar{\G}$ are connected in $\mathcal{D}$, denoted by $A \stackrel{\mathcal{D}}{\longleftrightarrow} B$, if there is a path connecting them which stays in $\mathcal{D}$.

For a set $S \subset \bar{\G}$, let $\mathcal{C}_{S}$ be the union of the connected components of $\mathcal{D}$ that intersect $S$.

For pairs of configurations, we denote by $\langle \cdot \rangle_{\G}^{A,+/-}$, or just $\langle \cdot \rangle^{A,+/-}$ if $\G$ is clear from context, the expectation with respect to the product measure $\bar{\P}_{\bar{G}}^{A,+} \otimes\bar{\P}_{\bar{G}}^{A,-}$, and by  $\left\langle \!\left\langle \cdot \right\rangle \!  \right\rangle$
the expectation with respect to the product of two identically distributed copies of $\bar{\P}_{\bar{\G}}$. The truncated correlation function $\langle \sigma_u ; \sigma_v\rangle$ denotes $\langle \sigma_u \sigma_v\rangle  - \langle \sigma_u\rangle\cdot \langle \sigma_v\rangle$.

\begin{prop}\label{prop:DisagreementRep}
For any random-field Ising model on a finite graph $\G$ and any realization of the external field $\eta$:
\begin{enumerate}
\item for all $u,v \in \vert(\G)$,
\be
\langle \sigma_u ; \sigma_v\rangle= 2 \cdot \left\langle \!\!\!\left\langle \1\left[u \stackrel{\mathcal{D}}{\longleftrightarrow} v\right]\right\rangle
\!\!\!\right\rangle.
\ee
\item For any  $A\subset  \vert(\G) $ and $u \in \vert(\G)$
\begin{eqnarray}\label{eq:DisagreementRepresentation}
\langle \sigma_u \rangle_{\G}^{A,+} &= &
\frac{1}{2} \cdot \left(
  \langle \sigma_u \rangle_{\G}^{A,+} -
  \langle \sigma_u \rangle_{\G}^{ A,-} \right)
  \notag \\[2ex]
  &=&
  \left\langle \1\left[u \stackrel{\mathcal{D}}{\longleftrightarrow} A\right]\right\rangle^{A,+/-} .
\end{eqnarray}
\end{enumerate}
\end{prop}
Applying equation \eqref{eq:DisagreementRepresentation}  with $\G = \Lambda(L) \subset \mathbb{Z}^2$ and $A = \extB \G$, we obtain the following geometric representation of the order parameter of equation \be \label{eq:order parameter}
m(L) =   \mathbb{E}\left[\left\langle \1\left[\zero \stackrel{\mathcal{D}}{\longleftrightarrow} \extB \Lambda(L)\right]\right\rangle^{\extB \Lambda(L),+/-} \right].
\ee

The identities of Proposition \ref{prop:DisagreementRep} are proven using a general principle of symmetry under the following class of \emph{swap} operation.  For any pairs of sets $A, S\subset \vert (\G)$,  let $R^A_S$ be the mapping
$R^A_S:  (\bar{\sigma},\bar{\sigma}') \mapsto   (\bar{\phi},\bar{\phi}')$ which acts as the identity in case $\mathcal{C}_S \cap A \neq \emptyset$, and otherwise swaps the configurations in $\mathcal{C}_S$, in the sense that
\be\label{eq:ClusterSwap}
(\bar\phi_u,\bar\phi'_u) := \begin{cases}  (\bar{\sigma}'_u,\bar{\sigma}_u) & u \in \mathcal{C}_S \\
(\bar{\sigma}_u,\bar{\sigma}'_u) & u \not \in \mathcal{C}_S \\ \end{cases}  \quad \mbox{(provided $\mathcal{C}_S \cap A = \emptyset$)}.
\ee

\begin{prop}\label{prop:swap}
For any $A,S \subset \vert(\G)$, and pair $\tau,\tau'$ of maps from $A$ to $\{-1,+1\}$,  the independent product measure $\bar{\P}_{\bar{\G}}^{A, \tau} \otimes\bar{\P}_{\bar{\G}}^{A, \tau'}$ (interpreted as $\bar{\P}_{\bar{\G}} \otimes\bar{\P}_{\bar{\G}}$ if $A=\emptyset$) is invariant under the action of $R_S^A$.
\end{prop}
\begin{proof}
Since $R_S^A$ is one-to-one (indeed, $R_S^A \circ R_S^A=\1$), it suffices to verify that it preserves the product of the weights of the two configurations.
That is so since, for each pair $(v,\e)$,  the factors in \eqref{eq:P_Lambda_tau_def} corresponding to the two configurations are either left unaffected or exchanged. The essential observation here is that the swap affects entire disagreement clusters, along whose external boundary the two configurations agree. \end{proof}
\noindent{\bf Remark:} The symmetry operation described in Proposition~\ref{prop:swap} is essentially equivalent to the `cluster swap' operation of Sheffield~\cite[Chapter 8]{S05} (with earlier uses by van den Berg in his disagreement percolation~\cite{B93}) applied to the (non-extended) RFIM (see~\cite[Section 2]{CAP17} for a review).

\begin{proof}[Proof of Proposition \ref{prop:DisagreementRep}]
The first item follows from the observation that
\be
\begin{split}
2\cdot \langle \sigma_u ; \sigma_v\rangle & = \left\langle \!\!\left\langle (\sigma_u - \sigma'_u) \cdot (\sigma_v - \sigma'_v)\right\rangle
\!\!\right\rangle  \\ & =\left\langle \!\!\!\left\langle (\sigma_u - \sigma'_u) \cdot (\sigma_v - \sigma'_v)\cdot \1\left[u \stackrel{\mathcal{D}}{\longleftrightarrow} v\right] \right\rangle
\!\!\!\right\rangle \\ & \quad \quad \quad+ \left\langle \!\!\!\left\langle (\sigma_u - \sigma'_u) \cdot (\sigma_v - \sigma'_v)\cdot \1\left[u \centernot{\stackrel{\mathcal{D}}{\longleftrightarrow}} v\right] \right\rangle
\!\!\!\right\rangle \\ & = \left\langle \!\!\!\left\langle (\sigma_u - \sigma'_u) \cdot (\sigma_v - \sigma'_v) \cdot \1\left[u \stackrel{\mathcal{D}}{\longleftrightarrow} v\right] \right\rangle
\!\!\!\right\rangle + 0,
\end{split}
\ee
where the final equality follows from the observation that, whenever $u$ and $v$ are disconnected in $\mathcal{D}$, we may `swap' the connected component of $u$ in $\mathcal{D}$ without affecting the value of $(\sigma_v,\sigma'_v)$ by Proposition \ref{prop:swap}. The proof is completed by noting that $(\sigma_u - \sigma'_u) \cdot (\sigma_v - \sigma'_v) = 4$ on the event $\{u \stackrel{\mathcal{D}}{\longleftrightarrow} v\}$.

The second item is proved analogously.
\end{proof}

\medskip

\section{The surface tension}  \label{sec:T}

\subsection{Definition and relation with disagreement percolation} \mbox{} \\[-1ex]

A useful role in the analysis is played by the following quantity.
\begin{definition} Given a finite graph $\G$ the surface tension between a pair of subset $A_1, A_2 \in  \vert(\G)$  at a given
 realization $\eta$, and at temperature $\beta^{-1}$ (which will often be omitted in our notation) is
\be \label{eq:def_T}
  \mathcal{T}_{A_1,A_2}(\eta, \beta)  :=   \frac{1}{\beta}  \cdot \log\left(
    \frac{Z^{+,+} \cdot Z^{-,-} } {Z^{+,-} \cdot Z^{-,+} }\right)\,.
\ee
where $ Z^{+,+} \equiv Z_\G^{A_1, A_2; s_1,s_2}$, at $s_1,s_2 =\pm$,   is the partition function with the spins restricted to the indicated values on $A_1$ and $A_2$.
\end{definition}

Of particular interest will be the situation where $A_j$ are the vertex boundaries of two subgraphs $ \Lambda_1 \subset \Lambda_2 \subset \G$, with  $\extB\Lambda_j = A_j$  (and $A_1$ and $A_2$ are disjoint).  In a slight abuse of notation, in such cases we shall refer to $  \mathcal{T}_{A_1,A_2}(\eta)$  also by the symbol
 $\mathcal{T}_{\Lambda_1,\Lambda_2}(\eta)$.  \\

The surface tension bears a number of relations with the disagreement percolation process
$\mathcal{D}(\bar{\sigma}^+,\bar{\sigma}^-)$, where $(\bar{\sigma}^+,\bar{\sigma}^-)$ are independently sampled from $\bar{\P}^{\extB \Lambda_2,+} \otimes \bar{\P}^{\extB\Lambda_2,-}$.
One of these is an identity, which appears in~\cite{AP18}, relating the surface tension to the thermal average of the number of disagreements in $\Lambda_1$. We define \be \label{eq:D_l_def}
  D_{\Lambda_1,\Lambda_2}(\eta)  \ := \frac{1}{2} \sum_{v\in \Lambda_1} \left[\langle \sigma_v \rangle_{\Lambda_2}^{+} - \langle \sigma_v \rangle_{\Lambda_2}^{-} \right],
 \ee
and observe that, by Proposition \ref{prop:DisagreementRep},
 \be\label{eq:disagreement representation for D}
D_{\Lambda_1,\Lambda_2}(\eta) = \langle |\Lambda_1 \cap \mathcal{C}_{\extB\Lambda_2}|  \rangle^{\extB \Lambda_2,+/-},
\ee
where we recall that $\mathcal{C}_{\extB \Lambda_2} := \mathcal{C}_{\extB \Lambda_2} (\bar{\sigma}^+,\bar{\sigma}^-)$
is the connected component of $\mathcal{D}$ containing $\extB \Lambda_2$. In these terms, one has:

\begin{thm} \label{thm:T2} Let $\Lambda_1 \subset \Lambda_2 \subset \G$ be subgraphs with disjoint internal vertex boundaries. For the RFIM with IID standard Gaussian random fields,
the surface tension bears  the following relation with disagreement percolation:
\be \label{DD2}
 \mathcal{T}_{\Lambda_1,\Lambda_2}(\eta)   \=  2 \eps\,
 \int_\R D_{\Lambda_1,\Lambda_2}(\eta^{(t)}) \, dt \ = \
 \ \frac{2 \eps  }{\sqrt{|\Lambda_1|}}\, \E _{\widehat \eta_{\Lambda_1} } \left(  \frac{D_{\Lambda_1,\Lambda_2}(\eta) }{ \phi(\widehat \eta_{\Lambda_1}) } \right)  \, ,
\ee
where: \begin{enumerate}[1)]
\item
 $\eta^{(t)}$ is defined by adding a uniform field of intensity $t$ in $\Lambda_1$
\begin{equation} \label{def:t}
  \eta^{(t)}_v := \begin{cases}
    \eta_v + t&v\in\Lambda_1\\
    \eta_v&\text{otherwise}
  \end{cases} \, ,
\end{equation}
\item the subscript on $ \E _{\widehat \eta_{\Lambda_1} } $ indicates the average over the variable
\be \label{eq:etahat}
\widehat \eta_{\Lambda_1}  := \frac 1{\sqrt{| \Lambda_1|}} \sum _{v\in  \Lambda_1} \eta_v   \,
\ee
at fixed values of the other, orthogonal, Gaussian degrees of freedom which determine $\eta$, \\

 \item  $\phi$ is the Gaussian density function $\phi(s) := \frac{1}{\sqrt{2\pi}}e^{-s^2/2}$.

 \end{enumerate}
\end{thm}
    \begin{figure}
\begin{center}
\includegraphics[width=0.25\textwidth]{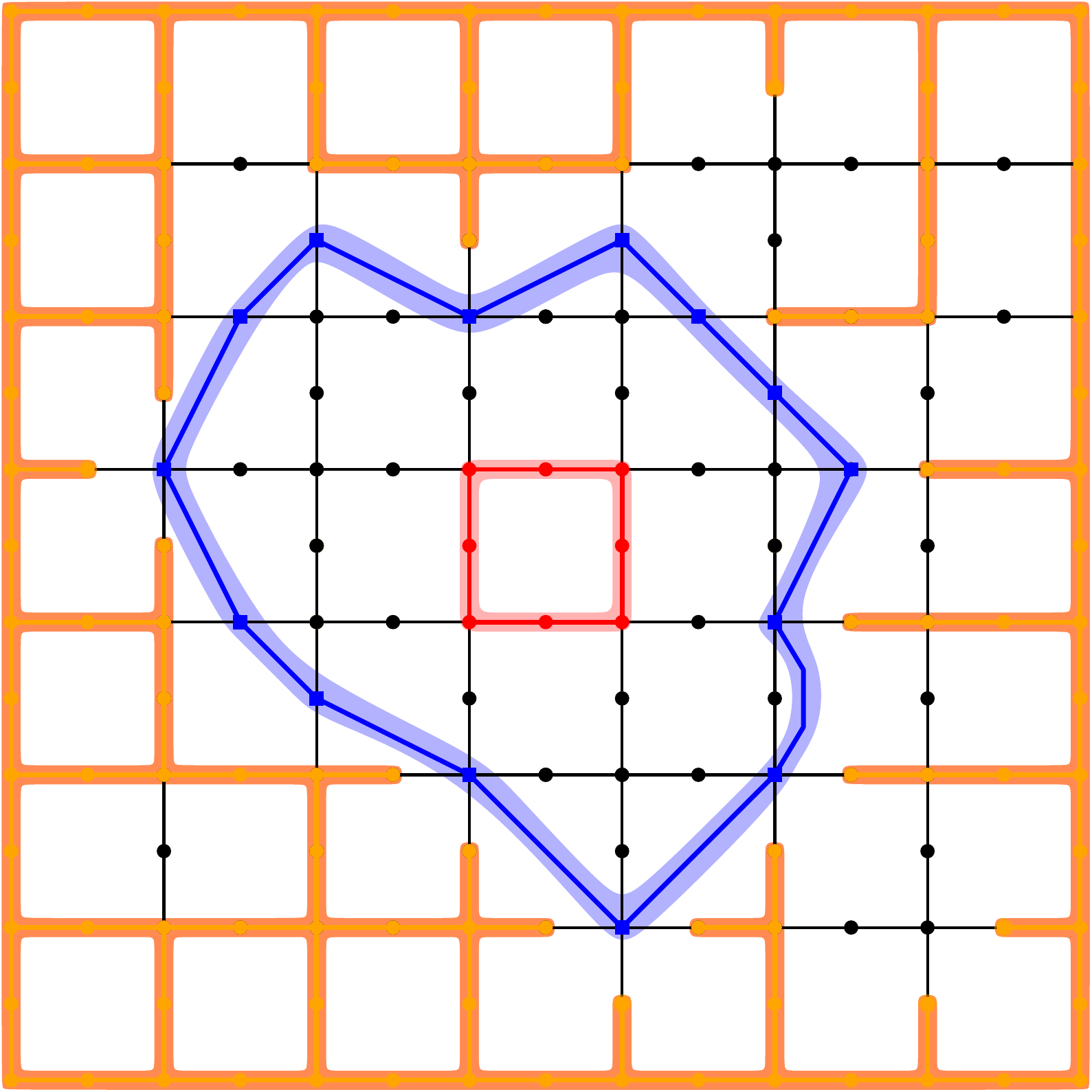}
\caption{Disagreement percolation in $\Lambda$, for a pair $(\sigma^+,\sigma^-)$.  The disagreement is marked in thick [orange] lines, along which $\bar \sigma^+> \bar \sigma^-$ and which connect to the set boundary $\partial_v \Lambda$.  The  loop [in blue] passes through sites at which $\bar \sigma^+ - \bar \sigma^- =0$, and  marks the outer vertex boundary  of the disagreement cluster of  $\extB \Lambda$. (Colors referring to the original PDF).}
\label{fig:disag}
\end{center}
\end{figure}
The proof given in~\cite{AP18} (in reference to $\Lambda_1 = \Lambda(\ell), \Lambda_2 = \Lambda(3\ell)$) extends to the presented statement with only notational modifications.\\

The second relation used here is an upper bound on $ \mathcal{T}_{\Lambda_1,\Lambda_2}(\eta)$.
For that, it  is instructive to first note the following identity.
\begin{prop}\label{prop:ExactRepST}
Let $\Lambda_1 \subset \Lambda_2 \subset \G$ such that $\extB \Lambda_1$ and $\extB \Lambda_2$ are disjoint. Then
\be \label{ZZ/ZZ}
\frac{\bar Z^{+,-} \cdot \bar Z^{-,+} } {\bar Z^{+,+} \cdot \bar Z^{-,-} } = \left\langle \1\left[\extB \Lambda_1 \stackrel{\mathcal{D}}{\centernot \longleftrightarrow} \extB \Lambda_2 \right] \right\rangle^{\extB\Lambda_1 \cup \extB \Lambda_2, +/-}.
\ee
\end{prop}

\begin{proof}
In this proof, we will denote by $\langle \cdot \rangle^{(s_1,s_2),(s_1',s_2')}$ the expectation operator induced by the product measure $\bar{\P}^{\extB \Lambda_2, \extB \Lambda_1, s_1,s_2} \otimes \bar{\P}^{\extB \Lambda_2, \extB \Lambda_1, s_1',s_2'}$.  In this notation
$$\langle \cdot \rangle^{(+,+),(-,-)} = \langle \cdot \rangle^{\extB\Lambda_1 \cup \extB \Lambda_2, +/-}\,.$$

Observe that
\be
\left\langle \1\left[\extB \Lambda_1 \stackrel{\mathcal{D}}{\centernot \longleftrightarrow} \extB \Lambda_2 \right] \right\rangle^{(+,-),(-,+)} =1 \, .
\ee
This is true because the pair $(\bar{\sigma},\bar{\sigma}')$ must `switch' from being identically $(+1,-1)$ on the boundary of $\Lambda_2$ to being identically $(-1,+1)$ on the boundary of $\Lambda_1$; in the extended Ising model, this is not compatible with a connection in $\mathcal{D}$. \\

Next, the symmetry argument which was used in the proof of Proposition \ref{prop:swap} implies also  that the product
\be
\bar{Z}^{+,s} \bar{Z}^{-,s'} \cdot \left\langle \1\left[\extB \Lambda_1 \stackrel{\mathcal{D}}{\centernot \longleftrightarrow} \extB \Lambda_2 \right] \right\rangle^{(+,s),(-,s')}
\ee
(which gives the contribution to the partition function of terms for which the indicated boundary conditions are met) is invariant under the exchange of $s$ with $s'$. From the resulting equality for the case $(s,s') = (-,+)$, we deduce that
\be
\begin{split}
\frac{\bar Z^{+,-} \cdot \bar Z^{-,+} } {\bar Z^{+,+} \cdot \bar Z^{-,-} } & = \frac{\left\langle \1\left[\extB \Lambda_1 \stackrel{\mathcal{D}}{\centernot \longleftrightarrow} \extB \Lambda_2 \right] \right\rangle^{(+,+),(-,-)}}{1} \\[2ex]  & \equiv \left\langle \1\left[\extB \Lambda_1 \stackrel{\mathcal{D}}{\centernot \longleftrightarrow} \extB \Lambda_2 \right] \right\rangle^{\extB\Lambda_1 \cup \extB \Lambda_2, +/-}. \qedhere
\end{split}
\ee
\end{proof}
\subsection{Bounding the surface tension via non-anticipatory sets}\label{sec:nonanticipatory}\mbox{} \\[-1ex]

The next upper bound combines a swap similar to that employed for \eqref{ZZ/ZZ}, while estimating the effects of imperfections that arise when the swap is not an exact symmetry.  It is an extension of Theorem 4.1 of~\cite{AP18} to random non-anticipatory sets, a notion which we now define.

Given $\Lambda_1 \subset \Lambda_2$ in $\bar{\G}$ and a set $S \subset \bar{\G}$, we denote the `forward' set of $S$ by
\be
S_{F} := \{ u \in \bar{\G} \setminus S: \exists \text{ a path from } u \text{ to } \extB \Lambda_1 \text{ which is disjoint from } S \}.
\ee
The complement of $S_F$ in $\bar{\Lambda}_2$ is called $S_B$, the `backward' set of $S$. We say that $S$ is a separating set if $\extB \Lambda_1 \in S_F$ and $\extB\Lambda_2 \in S_B$.

\begin{definition}  A random set $\mathcal{S}$ is called {\em non-anticipatory} if, for any set $S$, the event $\{\mathcal{S} = S\}$ is measurable with respect to the restriction of $(\bar{\sigma}^+,\bar{\sigma}^-)$ to $S_B$.
\end{definition}

\begin{prop}\label{prop:NonAnticipatory}
Let $\Lambda_1 \subset \Lambda_2 \subset \G$ be finite subsets, and fix a realization of the random field $\eta$. Let $\mathcal{S}$ be a non-anticipatory separating set such that, almost surely,
\be\label{eq:GoodSets}
\bar{\sigma}^+_{\e} = \bar{\sigma}^-_{\e} =0 \quad \forall \e \in \mathcal{S}.
\ee
Then
\begin{equation} \label{eq:Tpos_surface_tension}
\mathcal{T}_{\Lambda_1, \Lambda_2}(\eta)\leq 16 J\,\langle\, | \mathcal{S} \cap \mathcal{D} | \,\rangle^{\extB \Lambda_2 \cup \extB \Lambda_1,+/-}\,.
\end{equation}
\end{prop}

\begin{proof}
It suffices to prove the theorem for $\beta < \infty$, as both sides of \eqref{eq:GoodSets} are for almost every $\eta$ continuous at $\beta=\infty$.

Set $\mathcal{A}:=(\Lambda_2 \setminus \Lambda_1) \cup \extB \Lambda_1$.   As was explained in Lemma~\ref{lem:DMPExt} (cf. \eqref{Z_Zbar}) the partition function of a model can be  computed by summing the weights $\bar{\P}_{\bar{\G}}(\bar{\sigma})$  over the configurations of the extended version of the model.  The partial sums, corresponding to summing first over the edge variables, reduce the expression to the more standard sum of the  Gibbs weights $e^{-\beta H(\sigma)}$ summed  over the regular spin configurations  $\sigma$.  In the following proof we shall mix the two representations.

Given a deterministic separating set $S$ and an extended Ising configuration $\bar{\sigma}$, let $\bar{\sigma}_B$ and $\sigma_F$ be the restrictions of $\bar{\sigma}$ to $S_B$ and $S_F \cap \vert(\G)$, respectively.  Upon summing  the weights over the midedge variables $\e$ in the forward set, we are left with a mixed formula for the partition function in which the remaining summation is over the pair $(\bar \sigma_{B}, \sigma_{F})$, with suitably mixed weights.  If we further assume that $S$ satisfies \eqref{eq:GoodSets}, the weight takes on a particularly simple form of
\be
\bar{\P}_{S_B}(\bar{\sigma}_{B}) \cdot  e^{-\beta H(\sigma_F|\bar \sigma_S)},
\ee
where $H(\sigma_F|\bar \sigma_B)$ is the energy of $\sigma_F$ alone (including the contribution of the external field $\eta$ on $S_F \cap \vert(\G)$), plus the energy terms for pairs of neighboring vertices one of which is in $S_F \cap \vert(\G)$ and the other in $S_B \cap \vert(\G)$.  (In other words, vertex sites of $S$ act as an external boundary condition, and under the condition \eqref{eq:GoodSets} edge vertices of $S$ turn the edge into a free boundary condition.)

This prescription extends naturally to any non-anticipatory random set $\mathcal{S}$.  Decomposing the partition function according to the particular realization $S$ of the random set, one gets:
 \begin{multline} \label{Z_mixed}
\bar{Z}^{+,-} \cdot \bar{Z}^{-,+}  =\\
 = \sum_{S \subset \bar{\mathcal{A}}} \, \sum_{(\bar{\sigma}_B,\sigma_F),(\bar{\sigma}_B',\sigma_F')} \,
 \1[ \mathcal{S} = S](\bar{\sigma}_B,\bar{\sigma}'_B) \cdot
  \1^{\extB \Lambda_2, +}(\bar{\sigma}_B) \cdot \1^{\extB \Lambda_2, -} (\bar{\sigma}'_B)\cdot  \1^{\extB \Lambda_1, -}(\sigma_F)\cdot  \1^{\extB \Lambda_1, +}(\sigma'_F) \\[1.5ex]
  \quad \cdot W_S(\bar{\sigma}_B)\cdot  W_S(\bar{\sigma}'_B)
  \cdot  e^{-\beta [U_S(\bar{\sigma}_B) + U_S(\bar{\sigma}_B')]}
  \cdot \exp\left(-\beta [H_S(\sigma_F \mid \bar{\sigma}_S) + H_S(\sigma'_F \mid \bar{\sigma}'_S)] \right) \, , \qquad
\end{multline}
where the sum is over all separating sets $S$ and correspondingly split configurations, as explained above.  The first indicator ensures that the set $S$ coincides with the value of the random variable $\mathcal{S}$; the next four indicators impose the appropriate boundary conditions on $(\bar{\sigma},\bar{\sigma}')$.  The functions $W_S$ and $U_S$ represent the product of all weights $W$ over edges of $S_B$ and $U$ restricted to $S_B \cap \vert{\G}$, respectively. We note that since $\mathcal{S}$ is non-anticipatory the input to the first indicator can indeed be written as just the pair $(\bar{\sigma}_B,\bar{\sigma}'_B)$.

Upon the relabeling the variables $(\sigma_F,\sigma_F')$ in the switched order, \eqref{Z_mixed} transforms into:
\begin{multline}
\bar{Z}^{+,-} \cdot \bar{Z}^{-,+}   = \\
= \sum_{S \subset \bar{\mathcal{A}}} \,\sum_{(\bar{\sigma}_B,\sigma_F),(\bar{\sigma}_B',\sigma_F')} \1[ \mathcal{S} = S](\bar{\sigma}_B,\bar{\sigma}'_B) \cdot \1^{\extB \Lambda_2, +}(\bar{\sigma}_B) \cdot \1^{\extB \Lambda_2, -} (\bar{\sigma}'_B)\cdot  \1^{\extB \Lambda_1, +}(\sigma_F)\cdot  \1^{\extB \Lambda_1, -}(\sigma'_F)  \\[1.5ex]
\cdot W_S(\bar{\sigma}_B)\cdot  W_S(\bar{\sigma}'_B)    \cdot  e^{-\beta [U(\bar{\sigma}_B) + U(\bar{\sigma}_B')]}
\cdot \exp\left(-\beta [H_S(\sigma_F \mid \bar{\sigma}_S) + H_S(\sigma'_F \mid \bar{\sigma}'_S)+\right) \\[1.5ex]
\cdot   \exp\left(-\beta \cdot \Delta_S H (\sigma_F,\sigma'_F\mid \bar{\sigma}_B,\bar{\sigma}_B')  \right)\, ,  \quad
\end{multline}
with the energy \emph{switching  cost} function
\be
\begin{split}
\Delta_S H (\sigma,\sigma'\mid \tau,\tau') & := [H_S(\sigma \mid \tau) - H_S(\sigma'\mid \tau)] + [H_S(\sigma' \mid \tau') - H_S(\sigma \mid  \tau')] \\ & =  \sum_{\substack{(u,v) \in \edge(\G) \\ u \in S \cap \vert(\mathcal{A}), v \in S_F \cap \vert(\mathcal{A})}} J \cdot(\tau_u - \tau'_u) \cdot (\sigma_v-\sigma'_v).
\end{split}
\ee
The energy cost is null if there are no disagreement percolation sites along $\mathcal S$, and otherwise it is clearly  of the order of $| \mathcal {S}\cap \mathcal{D}|$.   For an explicit bound one easily gets:
\be \label{eq:DeltaH}
|\Delta_S H (\sigma,\sigma'\mid \tau,\tau') | \leq 16 J |\{ u\in S: \tau_u \neq \tau_u'\}|.
\ee

The above switch in the notation translates into a change in the boundary conditions, and can be summarized in the  following relation
\be
\begin{split}
\bar{Z}^{+,-} \cdot \bar{Z}^{-,+} & = \bar{Z}^{+,+} \cdot \bar{Z}^{-,-} \left\langle \exp(- \beta \cdot \Delta_{\mathcal{S}} H(\sigma_F,\sigma'_F \mid \sigma_\mathcal{S},\sigma'_{\mathcal{S}}) \right \rangle^{\extB \Lambda_2 \cup \extB \Lambda_1,+/-} \\
\end{split}
\ee
For the surface tension this gives:
\be
\begin{split}
\mathcal{T}_{\Lambda_1,\Lambda_2}(\eta) &  = \frac{-1}{\beta} \log \left\langle \exp(- \beta \cdot
\Delta_{\mathcal{S}} H(\sigma_F,\sigma'_F \mid \sigma_\mathcal{S},\sigma'_{\mathcal{S}}
\right \rangle^{\extB \Lambda_2 \cup \extB \Lambda_1,+/-}\\ & \leq  \frac{-1}{\beta} \log \left\langle \exp(- 16 \beta J |\mathcal{S} \cap \mathcal{D}| )\right \rangle^{\extB \Lambda_2 \cup \extB \Lambda_1,+/-} \\ & \leq 16 J \left\langle |\mathcal{S} \cap \mathcal{D}|\right \rangle^{\extB \Lambda_2 \cup \extB \Lambda_1,+/-},
\end{split}
 \ee
where the first inequality is \eqref{eq:DeltaH}, and the second is by Jensen's inequality applied to the convex function $(-\log x)$.
\end{proof}

\section{Fractality of disagreement percolation }

\subsection{A sufficient condition for tortuosity} \mbox{ } \\[-1ex]

The surface tension upper bound of the previous section will be significantly boosted through the observation that the path along the clusters of disagreement percolation are tortuous, in the following sense.
\begin{lemma}\label{lem:Fractality weaker}
Set $\mathcal{A}_{1,2} :=\Lambda(2 \ell) \setminus \Lambda(\ell)$. Let $A_{\alpha,\ell}$ denote the event that the annulus $\bar{\mathcal{A}}_{1,2}$ is crossed by a path of disagreement percolation  whose length is at most $\ell^{1 + \alpha}$. Then, for some $\alpha = \alpha(J/\eps)>0$:
\be
\lim_{\ell\to\infty}\mathbb{E}\left(\langle 1_{A_{\alpha,\ell}}\rangle^{\extB \mathcal{A}_{1,2},+/-}\right) \to 0.
\ee
\end{lemma}

The arguments of this section are close to the zero-temperature tortuosity argument of~\cite{DX19}. The desired statement follows from~\cite[Theorem 1.3]{AB99} through an estimate on the probabilities that collections of well-separated rectangles are simultaneously traversed in the long direction.

For an explicit statement of the condition, consider configurations $(\bar{\sigma}^+,\bar{\sigma}^-)$ defined in $\bar{\mathcal{A}}_{1,2}$, and the percolation cluster $\mathcal{C}_{\extB\mathcal{A}_{1,2}}(\bar{\sigma}^+,\bar{\sigma}^-)$ which they define.
We shall refer  by $\mathcal{R}$ to collections  of rectangles of possibly varying orientations in $\R^2$, with the properties:
\begin{enumerate}\label{list:well-separated}
  \item Each $R\in\mathcal{R}$ has side lengths $\ell(R)\times5\ell(R)$ with $10\le \ell(R)\le \frac{1}{160}\ell$.
  \item Each $R\in\mathcal{R}$ is contained in the annulus $\mathcal{A}_{\frac{5}{4},\frac{7}{4}}^{\R^2} :=\{x\in\R^2\colon \frac{5}{4}\ell<\|x\|_1< \frac{7}{4}\ell\}$.
  \item The $\ell_1(\R^2)$ distance between distinct $R_1, R_2\in\mathcal{R}$ is at least $60 \max\{\ell(R_1),\ell(R_2)\}$.
\end{enumerate}
We call such a collection {\em well-separated}.

We can associate $\mathcal{C}_{\extB\mathcal{A}_{1,2}}$ with a polygonal curve in $\mathbb{R}^2$ by embedding $\bar{\Z}^2$ in $\mathbb{R}^2$ in the natural way and including the straight lines between adjacent vertices in $\mathcal{C}_{\extB\mathcal{A}_{1,2}}$. We say that the event $\textup{Cross}(R)$ occurs if the curve associated to $\mathcal{C}_{\extB\mathcal{A}_{1,2}}$ contains a $\mathbb{R}^2$ path connecting the short sides of $R$.

By the general sufficient condition of ~\cite[Theorem 1.3]{AB99},  Lemma~\ref{lem:Fractality weaker} can be deduced from the following statement.
\begin{lemma}  \label{lem:criterion}
There exists  $b >0$ such that for any set of well-separated rectangles $\mathcal{R}$,
  \begin{equation} \label{eq:criterion}
    \E\left(\left\langle \prod_{R\in\mathcal{R}}\1_{\textup{Cross}(R)}\right\rangle^{\extB\mathcal{A}_{1,2}, +/-}\right)\le \left(1-b \right)^{|\mathcal{R}|}.
  \end{equation}
\end{lemma}
Our goal in this section is to prove this estimate, which we shall do with
\be
b \ = \  c\exp\left(-C\left(\frac{J}{\eps}\right)^2\right)
\ee
at some absolute constants $C, c>0$. This is followed by a quantitative version of Lemma~\ref{lem:Fractality weaker}.

\subsection{Disagreement lassoes}\mbox{ } \\[-1ex]

The first step is to bound the probability that the disagreement set forms a circuit in an annulus (the lasso event).   This would allow to  bound away from $1$ the crossing probability of a single rectangle, and finally verify  \eqref{eq:criterion}.\\

For this section we denote by $\L(\ell) $ the event that the boundary component of the disagreement set in $\bar{\Lambda}(25\ell)$ contains a circuit in the annulus $\overline{\Lambda(8\ell)\setminus\Lambda(\ell)}$, and refer to it as a \emph{lasso event}.\\

\begin{lemma}\label{lem:lasso}
  There exist absolute constants $C,c>0$ with which
  for all integers $\ell$
  \begin{equation}\label{eq:lasso estimate}
    \E(\langle \1_{\L(\ell)} \rangle_{\Lambda(25\ell)}^{\extB\Lambda(25\ell), +/-})\le 1-c\exp\left(-C\left(\frac{J}{\eps}\right)^2\right).
  \end{equation}
\end{lemma}

    \begin{figure}
\begin{center}
\includegraphics[width=0.4\textwidth]{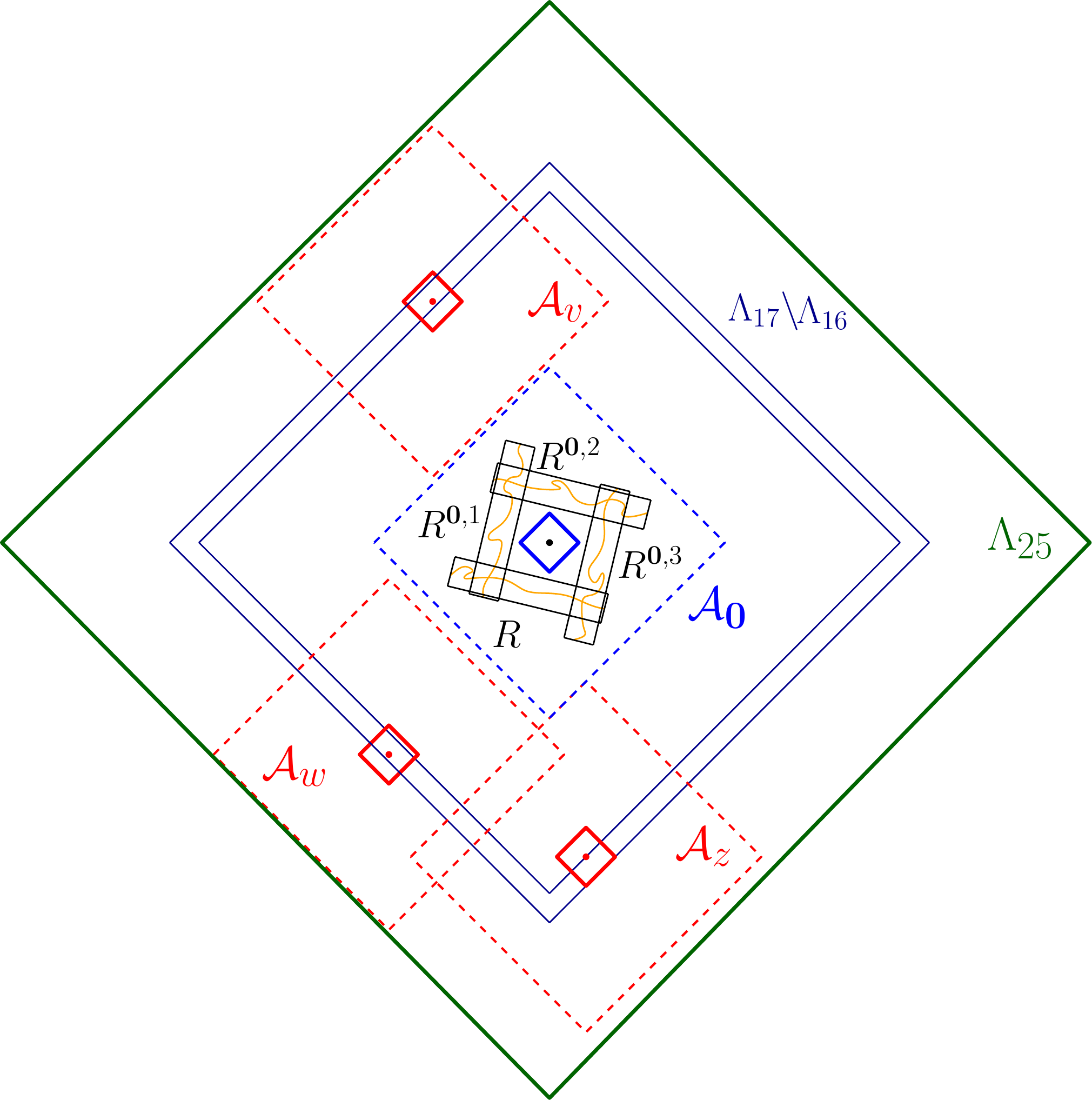}
\caption{The setup for the proof of Lemmas \ref{lem:criterion} and  \ref{lem:lasso}. The center point $\mathbf{0}$ is a pivot point for the retangle $R$. The annulus $\mathcal{A}_\mathbf{0}$ is in blue. The red annuli are translates of $\mathcal{A}_\mathbf{0}$ for $v,w,z \in \mathcal{B}$. Each red annulus is a subset of $\Lambda_{25} \setminus \Lambda_8$. (Colors referring to the original PDF).}
\label{fig:frac}
\end{center}
\end{figure}

\noindent We prove the lemma for $\beta < \infty$; the zero-temperature case follows by continuity.

A key tool towards the proof of  Lemma~\ref{lem:lasso} is a self-bounding  inequality that we now phrase.
Fix an integer $\ell>0$. For brevity, denote $\Lambda_j:=\Lambda(j\ell)$ and shorthand $\langle \cdot \rangle_{\Lambda_j}^{\extB\Lambda_j,+/-}$ to $\langle \cdot \rangle^{+/-}_j$.
Cover the annulus $\Lambda_{17}\setminus\Lambda_{16}$ by translates of $\Lambda_1$: $(v+\Lambda_1)_{v\in\mathcal{B}}$ for some $\mathcal{B}\subset \Lambda_{17}\setminus\Lambda_{16}$ chosen so that $|\mathcal{B}|$ smaller than a suitable absolute constant. Note that $v+\Lambda_8\subset\Lambda_{25}\setminus\Lambda_8$ for each $v\in\mathcal{B}$. Recalling~\eqref{eq:D_l_def} and~\eqref{eq:disagreement representation for D}, we denote the thermal average of the number of disagreements in $v+\Lambda_1$ due to boundary conditions placed at $v+\Lambda_8$ by
\begin{equation}\label{eq:D v eta}
  D_v(\eta):=D_{v+\Lambda_1, v+\Lambda_8}(\eta) = \left\langle|(v+\Lambda_1)\cap \mathcal{C}_{\extB(v+\Lambda_8)}|\right\rangle^{\extB(v+\Lambda_8), +/-}(\eta).
\end{equation}
\begin{lemma}\label{lem:self bounding inequality}
  For each realization of the random field,
  \begin{equation}\label{eq:self bounding inequality}
    \E _{\widehat \eta_{\Lambda_8} } \left(\frac{ D_\zero(\eta) \cdot  \left\langle \1_{ \L(\ell)} \right\rangle^{+/-}_{25}(\eta) }{ \phi(\widehat \eta_{\Lambda_8}) } \right) \le 100\cdot \frac{J}{\eps}\cdot \sum_{v\in\mathcal{B}} D_v(\eta).
  \end{equation}
\end{lemma}
\begin{proof}
Important aspects of the inequality include the fact that the \emph{same} distribution appears on both sides, namely $D_\zero$ on the left-hand side has the same distribution as each $D_v$ on the right-hand side, and the fact that the vector $(D_v)_{v\in\mathcal{B}}$ is \emph{independent} of $D_\zero$. The notation $\widehat \eta_{\Lambda}$ and $\phi$ are defined in Theorem~\ref{thm:T2}. The inequality is obtained by contrasting lower and upper bounds for the surface tension $\mathcal{T}_{\Lambda_8,\Lambda_{25}}$, as we now demonstrate.\\

From the integral representation of Theorem~\ref{thm:T2}, we see that
\begin{equation}\label{eq:surface tension first lower bound}
\mathcal{T}_{\Lambda_8,\Lambda_{25}}(\eta) = \frac{2 \eps  }{\sqrt{|\Lambda_8|}}\, \E _{\widehat \eta_{\Lambda_8} } \left(  \frac{D_{\Lambda_8,\Lambda_{25}}(\eta) }{ \phi(\widehat \eta_{\Lambda_8}) } \right).
\end{equation}
The right-hand side is further developed using~\eqref{eq:disagreement representation for D} by writing, for each realization of the external field,
\begin{equation}\label{eq:surface tension second lower bound}
  D_{\Lambda_8,\Lambda_{25}} \,=\, \left\langle |\Lambda_8 \cap \mathcal{C}_{\extB\Lambda_{25}}| \right\rangle^{+/-}_{25}\,\ge\, \left\langle |\Lambda_1 \cap \mathcal{C}_{\extB\Lambda_{25}}|\right\rangle^{+/-}_{25}\,\ge\,\left\langle |\Lambda_1 \cap \mathcal{C}_{\extB\Lambda_{25}}| \cdot  \1_{ \L}\right\rangle^{+/-}_{25}\,.
\end{equation}

The main purpose of introducing the lasso event $\L(\ell)$ is to define $+/-$ boundary conditions in the annulus $\Lambda_8\setminus \Lambda_1$ which decouple the disagreement set in $\bar{\Lambda}_1$ from the disagreement set outside $\bar{\Lambda}_8$. The event $\L(\ell)$ entails the existence of a self-avoiding path in $\mathcal{C}_{\extB\Lambda_{25}}$ which goes around the annulus $\overline{\Lambda_8\setminus\Lambda_1}$. The outermost such path $\mathcal{P}$ is well-defined (in the sense of the associated curves in $\mathbb{R}^2$, and is measurable with respect to the spin configurations $(\bar{\sigma}^+, \bar{\sigma}^-)$ on $\mathcal{P}$ and its exterior.
For each realization of the external field, the domain Markov property and monotonicity of the extended Ising model imply that
\begin{equation}\label{eq:surface tension third lower bound}
\begin{split}
  \left\langle |\Lambda_1 \cap \mathcal{C}_{\extB\Lambda_{25}}| \cdot  \1_{ \L(\ell)}\right\rangle^{+/-}_{25} &= \left\langle \left\langle|\Lambda_1 \cap \mathcal{C}_{\mathcal{P}}|\right\rangle^{\mathcal{P}, +/-} \cdot  \1_{ \L(\ell)} \right\rangle^{+/-}_{25} \\
  &\ge \left\langle|\Lambda_1 \cap \mathcal{C}_{\extB \Lambda_8}|\right\rangle^{ +/-}_8 \cdot  \left\langle \1_{ \L(\ell)} \right\rangle^{+/-}_{25}= D_\zero \cdot  \left\langle \1_{ \L(\ell)} \right\rangle^{+/-}_{25} \,.
\end{split}
\end{equation}

A complementary upper bound for the surface tension $\mathcal{T}_{\Lambda_1,\Lambda_{25}}$ is obtained by applying Proposition~\ref{prop:NonAnticipatory}. The disagreement percolation relevant to this context is the one induced by the product measure $\bar{\P}_{\bar{\Lambda}}^{\extB\Lambda,+} \otimes\bar{\P}_{\bar{\Lambda}}^{\extB{\Lambda,-}}$ with $\Lambda = \Lambda_{25}\setminus\Lambda(8\ell-1)$. For a given $8\ell\le k<25\ell$, consider the set $\mathcal{S}_k$ consisting of all $u \in \mathcal{D}^c \cup \extB \Lambda(k)$ where
there is a path of disgarrement from a neighbor of $u$ to $\extB \Lambda_{25}$ that stays within $\overline{\Lambda_{25}\setminus\Lambda(k-1)}$.

By definition, the number of disagreement vertices on $\mathcal{S}_k$ is at most $|\mathcal{C}_{\extB\Lambda_{25}}\cap \extB\Lambda(k)|$, and $\mathcal{S}_k$ is non-anticipatory. Furthermore, any $\e \in \mathcal{S}_k$ must have $\kappa^+_{\e} = \kappa^-_{\e} =0$. To see this, note that the two values of $\kappa$ must agree; since there must be a $v \in e$ satisfying $\sigma^+_v = +1$ and $\sigma^-_v = -1$, the two values of $\kappa_{\e}$ must be zero. Therefore, Proposition~\ref{prop:NonAnticipatory} implies that
\begin{equation}
  \mathcal{T}_{\Lambda_8,\Lambda_{25}}\le 16J\left\langle |\mathcal{C}_{\extB\Lambda_{25}}\cap \extB\Lambda_k|\right\rangle^{\extB\Lambda_{8}\cup\extB\Lambda_{25},+/-}\,.
\end{equation}
Averaging the obtained bound over $16\ell< k\le 17\ell$ implies that
\begin{equation}
  \mathcal{T}_{\Lambda_8,\Lambda_{25}}\le \frac{16J}{\ell}\left\langle |\mathcal{C}_{\extB\Lambda_{25}}\cap (\Lambda_{17}\setminus\Lambda_{16})|\right\rangle^{\extB\Lambda_{8}\cup\extB\Lambda_{25},+/-}\, .
\end{equation}
Lastly, recalling that $(v+\Lambda_1)_{v\in\mathcal{B}}$ covers $\Lambda_{17}\setminus\Lambda_{16}$, that $v+\Lambda_8\subset\Lambda_{25}\setminus\Lambda_8$ for each $v\in\mathcal{B}$ and the notation $D_v(\eta)$ from~\eqref{eq:D v eta} we may extend the last inequality to
\begin{equation}
  \mathcal{T}_{\Lambda_8,\Lambda_{25}}\le \frac{16J}{\ell}\sum_{v\in\mathcal{B}}D_v(\eta)\,
\end{equation}
where monotonicity of the model was again used. Putting the last inequality together with~\eqref{eq:surface tension first lower bound},~\eqref{eq:surface tension second lower bound},~\eqref{eq:surface tension third lower bound} and the fact that $\frac{16}{2}\cdot\frac{\sqrt{|\Lambda_8|}}{\ell}\le 100$ finishes the proof of Lemma~\ref{lem:self bounding inequality}.
\end{proof}

Next we deduce Lemma~\ref{lem:lasso} from the inequality of Lemma~\ref{lem:self bounding inequality}.

\begin{proof}[Proof of Lemma~\ref{lem:lasso}]  As the first step we make a judicious choice of a quantile value $q$ which will serve as a reference point for estimating both the left- and right-hand sides of inequality~\eqref{eq:self bounding inequality}. Define $S$ to be the normalized sum of the external field in the annulus $\Lambda_{25}\setminus\Lambda_8$,
\begin{equation}
  S:=\widehat \eta_{\Lambda_{25}\setminus\Lambda_8}=\frac{1}{\sqrt{|\Lambda_{25}\setminus\Lambda_8|}}\sum_{v\in\Lambda_{25}\setminus\Lambda_8}\eta_v.
\end{equation}
Let $q(S)$ be the $\frac{1}{2|\mathcal{B}|}$-quantile value of $D_v$ given $S$. That is,
\begin{equation}
  q(S) := \min\left\{x\colon\P(D_v>x\, |\, S)\le \frac{1}{2|\mathcal{B}|}\right\}.
\end{equation}
The definition remains the same regardless of the choice of $v\in\mathcal{B}$ as $(\eta_v)$ are independent and identically distributed and the variable $D_v$ depends on the external field only at $v+\Lambda_8$.
Lastly, let $q$ be the $(1 - \delta)$-quantile of $q(S)$, with $\delta$ a small positive number (defined by~\eqref{eq:rho value} and~\eqref{eq:r and delta 2 choice} below),
\begin{equation}
  q:=\min\left\{x\colon \P(q(S)>x)\le 1-\delta\right\}.
\end{equation}
To unravel the definitions, let
\begin{equation}
  \mathcal{S}^-:=\{s\colon q(s)\le q\},\quad \mathcal{S}^+:=\{s\colon q(s)\ge q\}
\end{equation}
and note that
\begin{equation}\label{eq:S- S+ prob}
  \P(S\in\mathcal{S}^-)\ge\delta,\quad \P(S\in\mathcal{S}^+)\ge 1-\delta
\end{equation}
and
\begin{alignat}{3}
  &\forall s\in\mathcal{S}^-, v\in\mathcal{B},\quad\quad &\P(D_v>q\,|\,S=s)\le\frac{1}{2|\mathcal{B}|},\\
  &\forall s\in\mathcal{S}^+, v\in\mathcal{B}, &\P(D_v\ge q\,|\,S=s)\ge\frac{1}{2|\mathcal{B}|}.\label{eq:S+ use}
\end{alignat}
In particular, by the union bound, the event $E := \{\sum_{v\in\mathcal{B}} D_v(\eta)\le q|\mathcal{B}|\}$ satisfies
\begin{equation}\label{eq:small sum for D}
  \P\left(E\right) \ge \P\left(E\cap\Big\{S\in\mathcal{S}^-\Big\}\right)\ge \frac{1}{2}\P(S\in\mathcal{S}^-)\ge \frac{1}{2}\delta.
\end{equation}
It is further noted that $q>0$ as, at positive temperature, $D_v(\eta)>0$ almost surely.

As the second step, we use the reference value $q$ and a suitably large positive parameter $r$ (defined by~\eqref{eq:r and delta 2 choice} below) to develop the inequality~\eqref{eq:self bounding inequality} to a form in which the lasso event is separated from the other terms:
\begin{equation}\label{eq:self bounding inequality simplified}
\begin{split}
   100\cdot \frac{J}{\eps}\cdot \sum_{v\in\mathcal{B}} D_v(\eta) &\ge q\cdot\E _{\widehat \eta_{\Lambda_8} } \left(\frac{ \1_{D_\zero(\eta)\ge q} \cdot  \left\langle \1_{ \L(\ell)} \right\rangle^{+/-}_{25}(\eta) }{ \phi(\widehat \eta_{\Lambda_8}) } \right)\\
   & \ge q\cdot\E _{\widehat \eta_{\Lambda_8} } \left(\frac{ \1_{D_\zero(\eta)\ge q} \cdot  \left\langle \1_{ \L(\ell)} \right\rangle^{+/-}_{25}(\eta)\cdot \1_{|\widehat \eta_{\Lambda_8}|\le r }}{ \phi(\widehat \eta_{\Lambda_8}) } \right)\\
   & \ge q\cdot\E _{\widehat \eta_{\Lambda_8} } \left(\frac{ \1_{D_\zero(\eta)\ge q} \cdot \1_{|\widehat \eta_{\Lambda_8}|\le r }}{ \phi(\widehat \eta_{\Lambda_8}) } \right) - q\cdot\E _{\widehat \eta_{\Lambda_8} } \left(\frac{ \left\langle \1_{\L(\ell)^c} \right\rangle^{+/-}_{25}(\eta)\cdot \1_{|\widehat \eta_{\Lambda_8}|\le r }}{ \phi(\widehat \eta_{\Lambda_8}) } \right)\\
   & \ge q\cdot\E _{\widehat \eta_{\Lambda_8} } \left(\frac{ \1_{D_\zero(\eta)\ge q} \cdot \1_{|\widehat \eta_{\Lambda_8}|\le r }}{ \phi(\widehat \eta_{\Lambda_8}) } \right) - \frac{q}{\phi(r)}\cdot\E _{\widehat \eta_{\Lambda_8} } \left(\left\langle \1_{\L(\ell)^c} \right\rangle^{+/-}_{25}(\eta)\right).
\end{split}
\end{equation}
As the inequality holds for all $\eta$, we may take its conditional expectation on the event $E$ and use the independence of $(D_0,\widehat \eta_{\Lambda_8})$ and $(D_v)_{v\in\mathcal{B}}$ to obtain
\begin{equation}
  100\cdot\frac{J}{\eps}\cdot q|\mathcal{B}| \, \ge \,q\cdot\E \left(\frac{ \1_{D_\zero(\eta)\ge q} \cdot \1_{|\widehat \eta_{\Lambda_8}|\le r }}{ \phi(\widehat \eta_{\Lambda_8}) } \right) - \frac{q}{\phi(r)}\E\left(\E _{\widehat \eta_{\Lambda_8} } \left(\left\langle \1_{\L(\ell)^c} \right\rangle^{+/-}_{25}(\eta)\right)\,|\,E\right)
\end{equation}
We can use~\eqref{eq:small sum for D} to simplify the right-most expectation above to
\begin{equation}
  \E\left(\E _{\widehat \eta_{\Lambda_8} } \left(\left\langle \1_{\L(\ell)^c} \right\rangle^{+/-}_{25}(\eta)\right)\,|\,E\right) \le \frac{1}{\P(E)} \E\left(\left\langle \1_{\L(\ell)^c} \right\rangle^{+/-}_{25}(\eta)\right)\le \frac{2}{\delta}\cdot\E\left(\left\langle \1_{\L(\ell)^c} \right\rangle^{+/-}_{25}(\eta)\right)\,.
\end{equation}
Putting together the last inequalities and rearranging (recalling that $q> 0$) shows that
\begin{equation}\label{eq:self bounding inequality after conditional expectation}
  \E\left(\left\langle \1_{\L^c} \right\rangle^{+/-}_{25}(\eta)\right)\ge \frac{1}{2}\cdot \delta\cdot \phi(r)\left(\E \left(\frac{ \1_{D_\zero(\eta)\ge q} \cdot \1_{|\widehat \eta_{\Lambda_8}|\le r }}{ \phi(\widehat \eta_{\Lambda_8}) } \right) - 100\cdot\frac{J}{\eps}\cdot |\mathcal{B}|\right)\,.
\end{equation}

As the third step, a lower bound is provided for the expectation in~\eqref{eq:self bounding inequality after conditional expectation}, making use of the specific choice of $q$. Start with the equality
\begin{equation}\label{eq:expectation to integral}
\begin{split}
  \E \left(\frac{ \1_{D_\zero(\eta)\ge q} \cdot \1_{|\widehat \eta_{\Lambda_8}|\le r }}{ \phi(\widehat \eta_{\Lambda_8}) }\right) &= \int_{-r}^r \P(D_\zero(\eta)\ge q\,|\,\widehat \eta_{\Lambda_8} = t)dt \\
  &=\int_{-r}^r \P(D_v(\eta)\ge q\,|\,\widehat \eta_{v+\Lambda_8} = t)dt,
\end{split}
\end{equation}
applying a translation by an arbitrary $v\in\mathcal{B}$. Observe that $\widehat \eta_{v+\Lambda_8}$ and $S$ are standard Gaussian random variables with correlation coefficient
\begin{equation}\label{eq:rho value}
  \rho := \E(\widehat \eta_{v+\Lambda_8}\cdot S) = \sqrt{\frac{|\Lambda_8|}{|\Lambda_{25}\setminus \Lambda_8|}}
\end{equation}
which is between two positive absolute constants. The required lower bound for the expression~\eqref{eq:expectation to integral} will be deduced from~\eqref{eq:S+ use}, which shows that for every~$s$,
\begin{equation}\label{eq:lower bound S+ rewrite}
\begin{split}
  \frac{1}{2|\mathcal{B}|}\1_{s\in\mathcal{S}^+}&\le \P(D_v(\eta)\ge q\,|\, S=s)\\
   &= \int\P(D_v(\eta)\ge q\,|\, S=s, \widehat \eta_{v+\Lambda_8} = t)\phi_\rho(t|s)dt\\
   &= \int\P(D_v(\eta)\ge q\,|\, \widehat \eta_{v+\Lambda_8} = t)\phi_\rho(t|s)dt,
\end{split}
\end{equation}
where we have written $\phi_\rho(\cdot | s)$ for the conditional density of $\widehat \eta_{v+\Lambda_8}$ given $\{S = s\}$, and where the fact that $D_v$ is independent of the external field outside $v+\Lambda_8$ is used. Two simple calculations involving the Gaussian density then show that
\begin{equation}\label{eq:first Gaussian calculation}
  s\in\left[-\frac{1}{\rho}r + c_1, \frac{1}{\rho}r - c_1\right]\quad\Longrightarrow\quad\int_{(-\infty,r)\cup(r,\infty)} \phi_{\rho}(t|s)dt \le \frac{1}{4|\mathcal{B}|}
\end{equation}
and
\begin{equation}\label{eq:second Gaussian calculation}
  \int \phi_{\rho}(t|s)ds = \frac{1}{\rho} \le C_1,
\end{equation}
where here and below we use $c_j, C_j$ to denote positive absolute constants. The first calculation allows to develop~\eqref{eq:lower bound S+ rewrite} to
\begin{equation}
  \int_{-r}^r\P(D_v(\eta)\ge q\,|\, \widehat \eta_{v+\Lambda_8} = t)\phi_\rho(t|s)dt\ge \frac{1}{4|\mathcal{B}|}\1_{s\in\mathcal{S}^+\cap\left[-\frac{1}{\rho}r + c_1, \frac{1}{\rho}r - c_1\right]}.
\end{equation}
The second calculation allows to integrate over $s$ in the last expression and obtain
\begin{equation}
  C_1 \int_{-r}^r\P(D_v(\eta)\ge q\,|\, \widehat \eta_{v+\Lambda_8} = t)dt\ge \frac{1}{4|\mathcal{B}|}\Leb\left(\mathcal{S}^+\cap\left[-\frac{1}{\rho}r + c_1, \frac{1}{\rho}r - c_1\right]\right)
\end{equation}
with $\Leb$ standing for Lebesgue measure. Comparing with~\eqref{eq:expectation to integral} we conclude that
\begin{equation}\label{eq:final lower bound on expectation}
  \E \left(\frac{ \1_{D_\zero(\eta)\ge q} \cdot \1_{|\widehat \eta_{\Lambda_8}|\le r }}{ \phi(\widehat \eta_{\Lambda_8}) }\right)\ge \frac{1}{4C_1\mathcal{|B|}}\Leb\left(\mathcal{S}^+\cap\left[-\frac{1}{\rho}r + c_1, \frac{1}{\rho}r - c_1\right]\right).
\end{equation}

The final step is to choose the parameters $r$ and $\delta$. Inequalities~\eqref{eq:self bounding inequality after conditional expectation} and~\eqref{eq:final lower bound on expectation} show that
\begin{equation}
  \E\left(\left\langle \1_{\L(\ell)^c} \right\rangle^{+/-}_{25}(\eta)\right)\ge \frac{1}{2}\cdot\delta\cdot\phi(r)\left(\frac{1}{4C_1\mathcal{|B|}}\Leb\left(\mathcal{S}^+\cap\left[-\frac{1}{\rho}r + c_1, \frac{1}{\rho}r - c_1\right]\right) - 100\cdot\frac{J}{\eps}\cdot |\mathcal{B}|\right)
\end{equation}
For any real subset $I$,
\begin{equation}
  \Leb(\mathcal{S}^+\cap I) = \Leb(I) - \Leb(I\setminus\mathcal{S}^+) \ge \Leb(I) - \frac{1}{\max\{\phi(x)\colon x\in I)\}}\P(S\notin \mathcal{S}^+).
\end{equation}
Thus, the fact that $\P(S\notin\mathcal{S}^+)\le \delta$ by~\eqref{eq:S- S+ prob} shows that taking
\begin{equation}\label{eq:r and delta 2 choice}
\begin{split}
  r &:= C_2\left(\frac{J}{\eps}+ 1\right)\\
  \delta &:= \phi\left(\frac{1}{\rho}\cdot r\right)
\end{split}
\end{equation}
with $C_2$ large suffices to ensure that
\begin{equation}
  \E\left(\left\langle \1_{\L(\ell)^c} \right\rangle^{+/-}_{25}(\eta)\right)\ge c_2\exp\left(-C_3\left(\frac{J}{\eps}\right)^2\right) \, ,
\end{equation}
thereby proving Lemma~\ref{lem:lasso}.
\end{proof}
\subsection{Verification of the tortuosity condition}
\begin{proof}[Proof of Lemma \ref{lem:criterion}]
Given a rectangle $R$ and a vertex $v \in \mathbb{Z}^2$, let $R^{v,j}$ be the rotation of $R$ by $j\pi/2$ radians around $v$. Let $\mathcal{A}_v$ be the set $\{x\in\R^2\colon \ell(R)< \|x-v\|_1\le 8\ell(R)\}$ (an $\R^2$-image of the annulus $\Lambda_v(8\ell(R))\setminus \Lambda_v(\ell(R))$). We say that $v$ is a pivot point for $R$ if i) $\cup_{j=0}^3 R^{v,j}$ is contained in $\mathcal{A}_v$, and ii) crossing all four rectangles (as in the event $\textup{Cross}_{R}$) implies the lasso event $\mathcal{L}(\ell)$ in the annulus (see Figure \ref{fig:frac} for an illustration).

Letting
\be
p_R :=\mathbb{E} \left[ \left\langle \1_{\textup{Cross}_R} \right\rangle^{\extB \Lambda_v(25 \ell(R)),+/-} \right],
\ee
rotation invariance, the union bound, and the definition of pivot points imply that
\be
1 - 4(1 - p_R) \leq \E\left[ \left\langle \1_{\cup_{j=0}^3 \textup{Cross}_{R^{v,j}}} \right\rangle^{\extB \Lambda_v(25 \ell(R)),+/-} \right] \leq \E\left[ \left\langle \1_{\mathcal{L}(\ell)} \right\rangle^{\extB \Lambda_v(25 \ell(R)),+/-} \right],
\ee
where $\mathcal{L}$ is the appropriate lasso event. By Lemma \ref{lem:lasso} and some algebraic manipulation, we deduce that, for some absolute $c, C> 0$,
\be
p_R \leq 1 - c \exp \left(-c \left(\frac{J}{\eps}\right)^2\right).
\ee

To complete the proof, let $\mathcal{R}$ be a well-separated set of rectangles, as in the statement of Lemma \ref{lem:criterion}. Let $\{v_R\}_{R \in \mathcal{R}}$ be chosen so that $v_R$ is a pivot point for $R$. By the conditions on $\mathcal{R}$ and the triangle inequality, the set $\{\mathcal{A}_{v_R}\}_{R \in \mathcal{R}}$ is made up of disjoint annuli. Therefore, the domain Markov property and monotonicity of the extended Ising model imply that
\be
    \E\left(\left\langle \prod_{R\in\mathcal{R}}\1_{\textup{Cross}(R)}\right\rangle^{\extB\mathcal{A}_{1,2}, +/-}\right)\le \prod_{R \in \mathcal{R}} p_R \leq \left[1 - c \exp \left(-c \left(\frac{J}{\eps}\right)^2\right) \right]^{|\mathcal{R}|},
\ee
as required. \end{proof}

\subsection{Quantified fractality bounds}  \mbox{ } \\[-1ex]

Using a standard argument of  percolation with finite-range dependence we next go beyond Lemma~\ref{lem:Fractality weaker}, extracting from it a quantified version of the statement.

\begin{thm}\label{thm:Fractality}
Let $A_{c,\alpha,\ell}$ be the event that the annulus
$\mathcal{A}_{1,2}:=\Lambda(2 \ell) \setminus \Lambda(\ell)$
is crossed by a path of disagreement percolation whose length does not exceed $c\cdot \ell^{1 + \alpha}$. Then there exist $\alpha = \alpha(J/\eps)>0$, absolute constants $C,c,c_0 >0$,  and $\ell_1 = \ell_1(J/\eps)>0$ such that for $\ell > \ell_1$
\be
\mathbb{E}\left(\langle 1_{A_{c_0,\alpha,\ell}}\rangle^{\extB \mathcal{A}_{1,2}(\ell),+/-}\right) \leq  C e^{-c \sqrt{\ell}}.
\ee
\end{thm}
\begin{proof}
For $v\in\Z^2$, define $\mathcal{A}_{v,s,t}(\ell):=(v+\Lambda(t \ell)) \setminus (v+ \Lambda(s \ell))$. Set $\alpha_0 = \alpha_0(J/\eps)>0$ to the $\alpha$ of Lemma~\ref{lem:Fractality weaker}. Define
\begin{equation}
  p(\ell) := \mathbb{E}\left(\langle 1_{A_{\alpha_0,\ell}}\rangle^{\extB \mathcal{A}_{1,2}(\ell),+/-}\right)
\end{equation}
so that $p(\ell)\to 0$ as $\ell\to\infty$ by Lemma~\ref{lem:Fractality weaker}. Below we use the convention that $C_j,c_j$ stand for positive absolute constants.

For integer $\ell>C_1$, consider the configuration $(\bar{\sigma}^+,\bar{\sigma}^-)$ sampled in the much larger annulus $\mathcal{A}_{1,2}(\ell^2)$. To get the desired quantitative bound, we will create an auxiliary percolation process on a rescaled version of $\mathbb{Z}^2$. This process will contain a crossing of an annulus whenever the original disagreement percolation contains a crossing of $\mathcal{A}_{1,2}(\ell^2)$ whose length is less than $\ell^{2 + \alpha_0}$. We will then prove an upper bound on the probability of crossing the rescaled annulus.

To that end, let
\be
\mathcal{V} = \{(m \ell, n \ell) : m,n \in \mathbb{Z}, \, m+n \text{ is even} \} \cap \mathcal{A}_{1.25,1.75}(\ell^2).
\ee
By construction, we have that
\be\label{eq:coveringbyannuli}
\mathcal{A}_{1.25,1.75}(\ell^2) \subset \bigcup_{v \in \mathcal{V}} \Lambda_v(\ell) \quad \text{and} \quad \forall v \in \mathcal{V}, \, \,  |\{w \in \mathcal{V}: d(v,w) \leq 4 \ell\} | < C_2.
\ee
We think of $\mathcal{V}$ as a `coarse' lattice, and endow it with a graph structure by saying that $(u,v) \in \edge(\mathcal{V})$ if $\Lambda_u(\ell)$ and $\Lambda_v(\ell)$ intersect along an edge (i.e. their intersection has two or more points). Viewed this way, $\mathcal{V}$ shares a graph structure with an annulus in $\mathbb{Z}^2$ of side length $c_1 \ell$ and some fixed aspect ratio. In a slight abuse of notation, we will call refer to this annulus by $\mathcal{V}$ as well. The construction ensures us that, if $P$ is a path in $\bar{\mathbb{Z}}^2$ which crosses $\bar{\mathcal{A}}_{1,2}(\ell^2)$ and $V(P) \subset \mathcal{V}$ is the set of vertices whose associated annuli $\mathcal{A}_{v,1,2}(\ell)$ are crossed by $P$, then $V(P)$ contains a crossing of the rescaled annulus $\mathcal{V}$.

Let $F_v$ be the event that $\mathcal{C}_{\extB \mathcal{A}_{1,2}(\ell^2)}(\bar{\sigma}^+, \bar{\sigma}^-)$ contains a crossing of $\mathcal{A}_{v,1,2}(\ell)$ of length at most $c_2 \ell^{1 + \alpha_0}$. We call the vertex $v$ `good' if $F_v$ occurs; otherwise, we call it `bad'. Then the event $A_{c_0,\alpha_0, \ell^2}$ implies that there must exist a crossing of $\mathcal{V}$ for which at least half the vertices are good. Letting $\mathcal{Q}$ be this event, the union bound gives that
\be\label{eq:unionbound}
\begin{split}
\mathbb{E}\left(\langle 1_{A_{c_0,\alpha_0,\ell^2}}\rangle^{\extB \mathcal{A}_{1,2}(\ell^2),+/-}\right) &\leq  \mathbb{E}\left(\langle 1_{\mathcal{Q}}\rangle^{\extB \mathcal{A}_{1,2}(\ell^2),+/-}\right)\\ &  \leq \sum_{\gamma} \sum_{\substack{S \subset \gamma \\ |S| \geq |\gamma|/2}} \mathbb{E}\left(\left\langle\prod_{v \in S} 1_{F_v}\right\rangle^{\extB \mathcal{A}_{1,2}(\ell^2),+/-}\right),
\end{split}
\ee
where $\gamma$ is summed over all possible crossing paths of the annulus $\mathcal{V}$, and $S$ is the set of good vertices in $\gamma$. The properties in \eqref{eq:coveringbyannuli} imply that, given any $S \subset \mathcal{V}$, there exists $S_0 \subset S$ such that $|S_0| \geq c_3\cdot|S|$, and $v,w \in S_0$ implies that $d(v,w) > 4 \ell$ -- i.e.~$\mathcal{A}_{v,1,2}(\ell) \cap \mathcal{A}_{w,1,2}(\ell) = \emptyset$. Therefore, the domain Markov property and monotonicity imply that
\begin{equation}
  \mathbb{E}\left(\left\langle \prod_{v\in S}1_{F_v}\right\rangle^{\extB \mathcal{A}_{1,2}(\ell^2),+/-}\right)\le p(\ell)^{c_3 \cdot |S|}.
\end{equation}
Noting that the shortest path crossing the annulus $\mathcal{V}$ has length $c_4 \ell$, \eqref{eq:unionbound} implies
\be\label{eq:unionbound2}
 \mathbb{E}\left(\langle 1_{\mathcal{Q}}\rangle^{\extB \mathcal{A}_{1,2}(\ell^2),+/-}\right) \leq \sum_{|\gamma| = c_4 \ell}^\infty 8^{|\gamma|} \cdot p(\ell)^{c_3 \cdot |\gamma|/2},
 \ee
since there are at most $4^{|\gamma|}$ paths of length $|\gamma|$, and at most $2^{|\gamma|}$ ways to partition the set into good and bad vertices. If $p(\ell) < c_5$, this sum is bounded above by $C_3 \cdot e^{-c_6 \ell}$. Setting $\alpha = \alpha_0/2$ and $\ell_1 = \ell_0^2$ for the minimal $\ell_0 > C_1$ for which $p(\ell) < c_5$, the proof is complete.
\end{proof}

\section{Exponential Decay}\label{sec:exp decay}
In this Section we prove Theorem~\ref{thm:exponential_bound}. As the main step, we show that the order parameter $m(\ell)$ decays faster than $\frac{c}{\ell}$ for all $c>0$.
\begin{prop}\label{prop:fast power law}
  For all $T$, $J$, $h$ and $\eps$,
  \begin{equation}\label{eq:fast power law}
    \lim_{\ell\to\infty} \ell\cdot m(\ell) = 0.
  \end{equation}
\end{prop}
A quantitative rate of decay for~\eqref{eq:fast power law}, which depends only on $J/\eps$ and is phrased in terms of the quantities resulting from the tortuosity theorem, is given in~\eqref{eq:quantitative dependence} below.

Theorem~\ref{thm:exponential_bound} follows from Proposition~\ref{prop:fast power law} by a standard percolation argument which was detailed in~\cite[Appendix A]{AP18} for the zero-temperature case.
The disagreement percolation representation allows its natural extension to positive temperatures, stated next. The proof is similar in spirit to the one used to prove Theorem \ref{thm:Fractality}.
\begin{prop} \label{thm:exp1}
For the RFIM on $\Z^d$ with the nearest-neighbor interaction, there is a finite constant $c_0$ (depending only on $d$) with which: if for some $\ell < \infty$
\be
m(\ell) \, \leq \,  c_0  / \ell^{d-1}\label{eq:one_over_ell_bound}
\ee
then for all $L <\infty$
\be
m(L) \, \leq \,  C_1 \,  e^{- b  L/ \ell}
\ee
with  $C_1, b \in (0,\infty)$ which do not depend on $T$, $J$, $h$, $\eps$ and $\ell$.
\end{prop}
\begin{proof}
By~\eqref{eq:order parameter},
\begin{equation}
m(L) =   \mathbb{E}\left(\left\langle \1\left[\zero \in\mathcal{C}_{\Lambda(L)}\right]\right\rangle^{\extB \Lambda(L),+/-} \right).
\end{equation}
Assume $L>4\ell$ (without loss of generality) and let $\mathcal{V}\subset\Lambda(\lfloor L/2\rfloor)$ be such that $(\Lambda_v(\ell))_{v\in \mathcal{V}}$ covers $\Lambda(\lfloor L/2\rfloor)$, and each $\Lambda_v(2\ell)$ intersects at most $C_2$ others, for a suitable $C_2$ depending only on $d$. Let $F_v$ be the event that $\Lambda_v(\ell)\cap \mathcal{C}_{\extB\Lambda(L)}\neq \emptyset$ (equivalently, $\extB\Lambda_v(\ell)\cap \mathcal{C}_{\extB\Lambda(L)}\neq \emptyset$). The domain Markov property, monotonicity and a union bound show that if $\mathcal{V}_0\subset\mathcal{V}$ are such that $(\Lambda_v(2\ell))_{v\in\mathcal{V}_0}$ are disjoint then
\begin{equation}
  \E\left(\left\langle \cap_{v\in\mathcal{V}_0}F_v\right\rangle^{\extB \Lambda(L),+/-}\right)\le \prod_{v\in\mathcal{V}_0} \E\left(\left\langle \extB\Lambda_v(\ell)\cap \mathcal{C}_{\extB \Lambda(2\ell)}\neq\emptyset\right\rangle^{\extB \Lambda(2\ell),+/-}\right)\le \left(C_3\cdot c_0\right)^{|\mathcal{V}_0|}
\end{equation}
with $C_3$ depending only on $d$.
Lastly, if $\zero\in\mathcal{C}_{\extB\Lambda(L)}$ then at least a constant times $L/\ell$ of the events $F_v$ occur, along a geometrically connected set of the $(\Lambda_v(\ell))_{v\in\mathcal{V}}$. Picking $c_0$ sufficiently small as a function of $d$, the result thus follows by standard arguments of percolation with finite-range dependence.
 \end{proof}
The rest of the section is devoted to proving Proposition~\ref{prop:fast power law}.

\subsection{Concentration and anti-concentration of the number of disagreements}\mbox{ }  \\[-2.5ex]

We first develop {\em a priori} bounds on the concentration properties of the number of disagreements. The ideas presented are adapted from corresponding ones in~\cite{AP18}.

The following inequality is used to convert upper bounds on the surface tension to an anti-concentration bound for the number of disagreements.
\begin{prop}\label{prop:AntiConcentration} For finite subgraphs $\Lambda_1\subset\Lambda_2\subset\Z^2$,
\[
\P\left[\frac{D_{\Lambda_1,\Lambda_2}}{\E[D_{\Lambda_1,\Lambda_2}]} < 1/2 \right] \geq \chi\left(\frac{1} {2 \eps }\cdot  \frac{ \E[\mathcal{T}_{\Lambda_1,\Lambda_2}]}{\sqrt{|\Lambda_1|}} \cdot \frac{|\Lambda_1|}{\E[D_{\Lambda_1,\Lambda_2}]} \right),
\]
where $\chi$ is the Gaussian distribution's two sided tail $\chi(t):=2\int_t^\infty \phi(s) \, ds$ (see Theorem~\ref{thm:T2} for the definition of $\phi$).
\end{prop}
We omit the proof since it is essentially identical to that of~\cite[Proposition 3.4]{AP18}.

The above anti-concentration bound will be contrasted with a conditional concentration inequality, which holds whenever the fast power-law decay~\eqref{eq:fast power law} of the order parameter is violated. The combination of the bounds results in a contradiction, which provides a proof for the fast decay~\eqref{eq:fast power law}.

The concentration inequality requires that the order parameter sequence $m(\ell)$ exhibits stretches with somewhat regular behavior. This is provided by the following abstract lemma.

\begin{lemma}\label{lem:comp_decay}
Let $(p_j)$ be a monotone non-increasing sequence satisfying $0\le p_j\le 1$. For each $\gamma>0$ and integer $k\ge 1$ there exists a non-negative integer~$n$ in the range
\begin{equation}\label{eq:inequalities on n}
(k+1)\cdot p_k^{1/(1+\gamma)}-1 \le n\le k
\end{equation}
such that for all $0\le j \leq n$,
\begin{equation}\label{eq:comp_dec}
  p_{n}\le p_j\le  p_{n}\left(\frac{n+1}{j+1}\right)^{1 + \gamma} \,.
\end{equation}
\end{lemma}
\begin{proof}
The left inequality in~\eqref{eq:comp_dec} follows from the fact that $(p_j)$ is non-increasing. The right inequality is obtained by selecting $n$ to be the index at which $(p_j \cdot (j+1)^{1 + \gamma})_{0 \leq j \leq k}$ is maximized. The lower bound in~\eqref{eq:inequalities on n} is due to $p_n\cdot (n+1)^{1+\gamma}\ge p_k\cdot (k+1)^{1+\gamma}$ and the assumption that $p_n \leq 1$.
\end{proof}

The next proposition contains the conditional concentration inequality discussed above.
\begin{prop}\label{prop:var_bound}
For each $c>0$ there exists $C = C(c)>0$ such that the following holds.
Let $\Lambda_1\subset\Lambda_2\subset\Z^2$ be finite subgraphs satisfying that for some $L\ge1$,
\begin{equation}\label{eq:geometric assumptions}
  |\Lambda_1|\ge c\cdot L^2\quad\text{and}\quad c\cdot L\le \min_{\substack{u\in\Lambda_1\\v\in\extB\Lambda_2}} d(u,v)\le \max_{\substack{u\in\Lambda_1\\v\in\extB\Lambda_2}} d(u,v)\le L.
\end{equation}
Assume that for some $0<\gamma<1$,
\begin{equation}\label{eq:comp decay assumption}
m(L) \leq m(j) \leq m(L) \left(\frac{L+1}{j+1}\right)^{1 + \gamma}, \quad 0 \leq j \leq L.
\end{equation}
Then
\begin{equation} \label{eq:Var_bnd}
  \E\big((D_{\Lambda_1,\Lambda_2})^2\big)\ \le\ \E\big(\big\langle |\Lambda_1\cap\mathcal{C}_{\extB \Lambda_2}|^2\big\rangle^{\extB \Lambda_2,+/-}\big) \  \le \ \frac{C}{\gamma}\cdot L^{2 \gamma} \cdot \big(\E \left( D_{\Lambda_1,\Lambda_2}\right) \big)^{2}.
\end{equation}
\end{prop}
The proposition will be applied once when $\Lambda_1$ and $\Lambda_2$ are concentric graph balls and once when they are concentric annuli.
\begin{proof}[Proof of Proposition \ref{prop:var_bound}]
Let $I_u$ be the indicator function of the event $\{u \in \mathcal{C}_{\extB\Lambda_2}(\bar{\sigma}^+,\bar{\sigma}^-)\}$ for $u \in \Lambda_1$. Proposition~\ref{prop:DisagreementRep} shows that
\begin{equation}
  D_{\Lambda_1,\Lambda_2} = \sum_{u \in \Lambda_1} \langle I_u\rangle^{\extB \Lambda_2 ,+/-}.
\end{equation}
The first inequality in~\eqref{eq:Var_bnd} thus follows from the Cauchy-Schwartz inequality and we focus on the second inequality. The assumption~\eqref{eq:geometric assumptions} and monotonicity imply that
\be\label{eq:FirstMomentLB}
\mathbb{E}(D_{\Lambda_1,\Lambda_2}) = \sum_{u \in \Lambda_1} \mathbb{E}\left(\langle I_u\rangle^{\extB \Lambda_2 ,+/-}\right) \geq |\Lambda_1|\cdot m(L) \ge c\cdot L^2\cdot m(L).
\ee
To bound the second moment, write
\begin{equation}
\E\left(\left\langle |\Lambda_1\cap\mathcal{C}_{\extB\Lambda_2}|^2\right\rangle^{\extB \Lambda_2,+/-}\right) = \sum_{u,v \in \Lambda_1} \mathbb{E}\left(\langle I_u \cdot I_v\rangle^{\extB \Lambda_2,+/-}\right)
\end{equation}
and apply the following estimate, based on the domain Markov property and monotonicity, to each term:
\begin{equation}\label{eq:far u v estimate}
\begin{split}
   \mathbb{E}\left(\langle I_u \cdot I_v\rangle^{\extB \Lambda_2,+/-}\right)&\leq \mathbb{E}\left(\langle I_u \rangle^{\extB \Lambda_u(r(u,v)),+/-} \cdot \langle I_v \rangle^{ \extB \Lambda_v(r(u,v)),+/-} \right)\\
&= \mathbb{E}\left(\langle I_u \rangle^{\extB \Lambda_u(r(u,v)),+/-}\right) \cdot \mathbb{E}\left(\langle I_v \rangle^{ \extB \Lambda_v(r(u,v)),+/-} \right)= m\left(r(u,v)\right)^2
\end{split}
\end{equation}
where $r(u,v) := \min\{\lfloor (d(u,v) - 1) /2\rfloor, \lfloor c\cdot L\rfloor\}$ is chosen so that the thermal expectations for $I_u$ and $I_v$ in~\eqref{eq:far u v estimate} are taken in disjoint subregions of $\Lambda_2$, and thus involve independent external fields.

Putting the last bounds together we conclude that
\begin{equation}\label{eq:second moment bound}
\begin{split}
  \E\left(\left\langle |\Lambda_1\cap\mathcal{C}_{\extB\Lambda_2}|^2\right\rangle^{\extB \Lambda_2,+/-}\right) &\le \sum_{k=0}^{\lfloor c\cdot L\rfloor} |\{u,v\in\Lambda_1\colon r(u,v) = k\}|\cdot m(k)^2\\
  &\le C_1|\Lambda_1|\left(\sum_{k=0}^{\lfloor c\cdot L\rfloor -1} (k+1)\cdot m(k)^2 + |\Lambda_1|\cdot m(\lfloor c\cdot L\rfloor)^2\right)
\end{split}
\end{equation}
where here and below we use $C_j$ to denote constants which may depend only on $c$. The resulting terms may be estimated via the assumption~\eqref{eq:comp decay assumption} as
\begin{equation}
  \sum_{k=0}^{\lfloor c\cdot L\rfloor -1} (k+1)\cdot m(k)^2\le m(L)^2\cdot (L+1)^{2+2\gamma}\cdot\sum_{k=0}^\infty \frac{1}{(k+1)^{1+2\gamma}}\le \frac{C_2}{\gamma} \cdot L^{2+2\gamma}\cdot m(L)^2
\end{equation}
and
\begin{equation}
  m(\lfloor c\cdot L\rfloor) \le C_3 \cdot m(L).
\end{equation}
Plugging these bounds in~\eqref{eq:second moment bound} and using~\eqref{eq:geometric assumptions} shows that
\begin{equation}
  \E\left(\left\langle |\Lambda_1\cap\mathcal{C}_{\extB\Lambda_2}|^2\right\rangle^{\extB \Lambda_2,+/-}\right)\le \frac{C_4}{\gamma}\cdot L^{4+2\gamma}\cdot m(L)^2.
\end{equation}
The result now follows by comparing with the lower bound~\eqref{eq:FirstMomentLB}.
\end{proof}

\subsection{An upper bound on the surface tension} \mbox{ } \\[-1ex]

The tortuosity bounds of Theorem \ref{thm:Fractality} can be used to construct non-anticipatory sets with relatively few points of disagreement, providing upper bounds on the surface tension via Proposition \ref{prop:NonAnticipatory}. These improve upon the bounds arising from deterministic sets by a power of the length scale (under the assumption~\ref{eq:RegulariationL2}).

\begin{prop}\label{prop:GoodSTUpperBound}
Let $\alpha = \alpha(J/\eps)$ and $\ell_1 = \ell_1(J/\eps)$ be the constants of Theorem \ref{thm:Fractality}. Let $0<\gamma<1$ and let $L>\ell_1$ be an integer satisfying that
\be\label{eq:RegulariationL2}
m(L) \leq m(j) \leq m(L) \left(\frac{L+1}{j+1}\right)^{1+ \gamma} \quad 0 \leq j \leq L.
\ee
Set $\ell = \lfloor L/ 5 \rfloor$. Then there exists a universal constant $C$ such that
\be
\mathbb{E}[\mathcal{T}_{\Lambda(\ell) ,\Lambda(6\ell)}] \leq C\gamma^{-\frac{2}{3}}\cdot J \cdot m(\ell) \cdot  \ell^{1 - \frac{1}{3}\alpha + \gamma}.
\ee
\end{prop}
The proof is divided into several lemmas. Recall the notation $S_B$ from Section~\ref{sec:nonanticipatory}.

\begin{lemma}\label{lem:nested non anticipatory}
  Let $\Lambda_1\subset\Lambda_2\subset\Z^2$ be finite subgraphs. Let $(\mathcal{S}_n)_{n=1}^N$ be a sequence of non-anticipatory separating sets which is increasing in the sense that
  \begin{equation}\label{eq:increasing assumption}
    (\mathcal{S}_n)_B\subset (\mathcal{S}_{n+1})_B,\quad 1\le n\le N-1.
  \end{equation}
  and with each $\mathcal{S}_n$ satisfying the assumption~\eqref{eq:GoodSets} almost surely. Then for each $M>0$,
  \begin{equation}\label{eq:surface tension bound}
    \mathcal{T}_{\Lambda_1 ,\Lambda_2}\le 16J\left(M + \langle\, | \mathcal{S}_N \cap \mathcal{D} |\cdot \1_E \,\rangle^{\extB \Lambda_2 \cup \extB \Lambda_1,+/-}\right)\,
  \end{equation}
  where $\mathcal{D}$ is the disagreement set~\eqref{eq:disagreement representation for D}  and $E$ is the event that $|\mathcal{S}_n \cap \mathcal{D}|> M$ for $1\le n\le N$.
\end{lemma}
\begin{proof}
  Let $1\le \mathcal{N}\le N$ be the smallest integer satisfying $|\mathcal{S}_\mathcal{N} \cap \mathcal{D}|\le M$ if such an integer exists (i.e., on the event $E^c$) and otherwise set $\mathcal{N}=N$. The increasing property of $(\mathcal{S}_n)$ then implies that $\mathcal{S}_\mathcal{N}$ is also non-anticipatory. The result then follows from Proposition~\ref{prop:NonAnticipatory}.
\end{proof}
We proceed to adapt the bound of the previous lemma to the setting of Proposition~\ref{prop:GoodSTUpperBound}, which requires several definitions.

We work in the annulus $\Lambda(6\ell)\setminus\Lambda(\ell-1)$ for some $\ell>0$. For brevity, we set
\begin{equation}
 \langle\cdot\rangle:=\langle\cdot\rangle^{\extB \Lambda(6\ell) \cup \extB \Lambda(\ell),+/-}.
\end{equation}
Let $2\ell< k\le3\ell$ be an integer. Set $\mathcal{A}_{k,6\ell}:=\overline{\Lambda(6\ell)}\setminus\overline{\Lambda(k)}$ (which does not include $\Lambda(k)$ but includes the midedges of the edges connecting $\Lambda(k)$ and $\Lambda(k+1)$). Define
\begin{equation}
  \mathcal{C}_{k|6\ell}:=\{u\in \mathcal{A}_{k,6\ell}\colon u\xleftrightarrow{\mathcal{D}\cap\mathcal{A}_{k,6\ell}}\extB\Lambda(6\ell)\}.
\end{equation}
Write $d_{\mathcal{C}_{k|6\ell}}$ for the graph distance in the induced subgraph on $\mathcal{C}_{k|6\ell}\subset\bar{\Z}^2$. Let $N_k:=|\mathcal{A}_{k,6\ell}|$ and define, for $1\le n\le N_k$,
\begin{equation}
  \begin{split}
    \mathcal{B}_n^k&:=\{u\in\mathcal{A}_{k,6\ell}\colon d_{\mathcal{C}_{k|6\ell}}(u,\extB\Lambda(6\ell))\le 2n-1\}\,,\\
    \mathcal{S}_n^k&:=\{u\in\overline{\Lambda(6\ell)}\setminus\mathcal{B}_n\colon \exists v\in\mathcal{B}_n, u\sim v\}.
  \end{split}
\end{equation}
Lastly, let
\begin{equation}
  \mathcal{C}_{6\ell}:=\{u\in \Lambda(6\ell)\setminus\Lambda(2\ell)\colon u\stackrel{\mathcal{D}}\longleftrightarrow\extB\Lambda(6\ell)\}.
\end{equation}
\begin{lemma}\label{lem:surface tension bound1}
  With the above definitions, for each $M>0$,
  \begin{equation}\label{eq:surface tension bound1}
    \mathcal{T}_{\Lambda(\ell) ,\Lambda(6\ell)}\le 16J\left(M + \frac{1}{\ell}\sum_{k=2\ell+1}^{3\ell}\langle\, | \mathcal{C}_{6\ell} \cap \Lambda(k) |\cdot \1_{E^k} \,\rangle\right)\,
  \end{equation}
  with $E^k$ the event that $|\mathcal{S}_n^k \cap \mathcal{C}_{6\ell}|> M$ for $1\le n\le N_k$.
\end{lemma}
\begin{proof}
We aim to apply Lemma~\ref{lem:nested non anticipatory} with $\Lambda_1 := \Lambda(\ell)$, $\Lambda_2 := \Lambda(6\ell)$.
The definitions imply that the $(S_n^k)$ are non-anticipatory separating sets satisfying~\eqref{eq:increasing assumption}. Moreover, each $\e \in \mathcal{S}_n^k$ satisfies that one endpoint of $e$ is in $\mathcal{B}_n^k$ while $\e\notin\mathcal{B}_n^k$. The use of an odd distance in the definition of $\mathcal{B}_n^k$ thus necessitates that $\bar{\sigma}^+|_{\e} = \bar{\sigma}^-|_{\e} =0$, so that the $\mathcal{S}_n^k$ satisfy the assumption~\eqref{eq:GoodSets} almost surely. Lemma~\ref{lem:nested non anticipatory} then shows that
\begin{equation}\label{eq:surface tension bound step}
  \mathcal{T}_{\Lambda(\ell) ,\Lambda(6\ell)}\le 16J\left(M + \langle\, | \mathcal{S}_{N_k}^k\cap \mathcal{D} |\cdot \1_{E^k} \,\rangle\right)\,
\end{equation}
with $E^k$ the event that $|\mathcal{S}_n^k \cap \mathcal{D}|> M$ for $1\le n\le N_k$. The event $E^k$ may be equivalently presented as in the statement of the lemma as our definitions imply that $\mathcal{S}_n^k \cap \mathcal{D} = \mathcal{S}_n^k \cap \mathcal{C}_{6\ell}$. In addition, the definitions and our choice of $N_k$ imply that $\mathcal{S}_{N_k}\cap\mathcal{D}\subset \mathcal{C}_{6\ell}\cap\Lambda(k)$. The lemma thus follows by averaging~\eqref{eq:surface tension bound step} over $2\ell<k\le 3\ell$.
\end{proof}

We next discuss the `bad' events $E^k$ appearing in Lemma~\ref{lem:surface tension bound1}. Let $\alpha = \alpha(J/\eps)$, $\ell_1 = \ell_1(J/\eps)$ and $c_0$ be the constants of Theorem \ref{thm:Fractality}. Define
\begin{equation}
  \begin{split}
    &E_1 := \left\{\text{$\mathcal{C}_{6\ell}$ contains a crossing of $\overline{\Lambda(6\ell)\setminus\Lambda(3\ell)}$ of length at most $c_0\cdot (3\ell)^{1+\alpha}$}\right\},\\
    &E_2 := \left\{|\mathcal{C}_{6\ell}\cap (\Lambda(6\ell)\setminus\Lambda(3\ell))|>\frac{1}{2}c\cdot M\cdot (3\ell)^{1+\alpha}\right\}.
  \end{split}
\end{equation}
\begin{lemma}
  With the above definitions, $E^k \subset E_1\cup E_2$ for each $2\ell\le k<3\ell$.
\end{lemma}
\begin{proof}
  Fix $2\ell\le k<3\ell$. The definitions imply that the sets $(\mathcal{S}_n^k\cap \mathcal{C}_{6\ell})_n$ are subsets of the non-extended lattice $\Z^2$, which may intersect only at $\Lambda(k)$. In particular,
  \begin{equation}
    |\mathcal{C}_{6\ell}\cap(\Lambda(6\ell)\setminus\Lambda(3\ell))|\ge \sum_{n=1}^{N_k}|\mathcal{S}_n^k\cap \mathcal{C}_{6\ell}\cap (\Lambda(6\ell)\setminus\Lambda(3\ell))|
  \end{equation}
  On the event $E^k$, this estimate may be continued to
  \begin{equation}
  \begin{split}
    |\mathcal{C}_{6\ell}\cap(\Lambda(6\ell)\setminus\Lambda(3\ell))| &> M \max\{n\colon (\mathcal{S}_n^k\cap\mathcal{C}_{6\ell}\subset\Lambda(6\ell)\setminus\Lambda(3\ell)\}\\
    &\ge M \max\{n\colon B_n^k\subset\Lambda(6\ell)\setminus\Lambda(3\ell)\}
  \end{split}
  \end{equation}
  where in the second inequality the fact that the $\mathcal{S}_n^k$ satisfy~\eqref{eq:GoodSets} is used. If the event $E_1$ does not occur then $B_n^k\subset\Lambda(6\ell)\setminus\Lambda(3\ell)$ for $1\le n\le \frac{1}{2}c\cdot (3\ell)^{1+\alpha}$. Thus $E^k\setminus E_1 \subset E_2$.
\end{proof}
A combination of the last two lemmas provides the estimate
\begin{equation}\label{eq:surface tension bound2}
  \mathcal{T}_{\Lambda(\ell) ,\Lambda(6\ell)}\le 16J\left(M + \frac{1}{\ell}\left\langle\, \left| \mathcal{C}_{6\ell} \cap (\Lambda(3\ell)\setminus\Lambda(2\ell)) \right|\cdot \1_{E_1\cup E_2} \,\right\rangle\right)\,.
\end{equation}
To further estimate the right-hand side we take the average over the external magnetic field and apply the Cauchy-Schwartz inequality and the union bound to get
\begin{multline}\label{eq:Cauchy-Schwartz}
\mathbb{E}\left(\left\langle\, \left| \mathcal{C}_{6\ell} \cap (\Lambda(3\ell)\setminus\Lambda(2\ell)) \right|\cdot \1_{E_1\cup E_2} \,\right\rangle\right)\\
\le\sqrt{\mathbb{E}\left(\left\langle\, \left| \mathcal{C}_{6\ell} \cap (\Lambda(3\ell)\setminus\Lambda(2\ell)) \right|^2\right\rangle\right)\cdot \mathbb{E}\Big(\Big\langle\1_{E_1}+\1_{E_2} \,\Big\rangle\Big)}
\end{multline}
We apply the following bounds to the terms in the right-hand side above.
\begin{lemma}\label{lem:error terms}
  Suppose $\ell = \lfloor L/5\rfloor$ for an integer $L$ satisfying~\eqref{eq:RegulariationL2} and suppose that $3\ell>\ell_1$. Then, for absolute constants $C,c'>0$,
  \begin{align}
    &\mathbb{E}\left(\left\langle\, \left| \mathcal{C}_{6\ell} \cap (\Lambda(3\ell)\setminus\Lambda(2\ell)) \right|^2\right\rangle\right)\le \frac{C}{\gamma}\cdot m(\ell)^2\cdot \ell^{4+2\gamma}\,,\label{eq:first error bound}\\
    &\mathbb{E}\left(\left\langle \1_{E_1}\right\rangle\right)\le C e^{-c' \sqrt{\ell}}\,,\label{eq:second error bound}\\
    &\mathbb{E}\left(\left\langle \1_{E_2}\right\rangle\right)\le \frac{C}{\gamma}\cdot \frac{m(\ell)\cdot \ell^{1+\gamma-\alpha}}{M}\,.\label{eq:third error bound}
  \end{align}
\end{lemma}
\begin{proof}
  To see~\eqref{eq:first error bound} note that Proposition~\ref{prop:var_bound} shows that
  \begin{equation}
    \mathbb{E}\left(\left\langle\, \left| \mathcal{C}_{6\ell} \cap (\Lambda(3\ell)\setminus\Lambda(2\ell)) \right|^2\right\rangle\right)\le\frac{C_1}{\gamma}\cdot \ell^{2\gamma} \cdot \big(\E \left( D_{\Lambda(3\ell)\setminus\Lambda(2\ell),\Lambda(6\ell)\setminus\Lambda(\ell-1)}\right) \big)^{2}.
  \end{equation}
  for an absolute constant $C_1$. The bound~\eqref{eq:first error bound} now follows by monotoncity of the disagreement percolation and the assumption~\eqref{eq:RegulariationL2}.

  The probability of the event $E_1$ is estimated using Theorem~\ref{thm:Fractality} and monotonicity. For the event $E_2$, monotonicity and the assumption~\eqref{eq:RegulariationL2} may again be invoked to yield
  \begin{equation}\label{eq:D estimate}
  \begin{split}
    \mathbb{E}\left(\left\langle|\mathcal{C}_{6\ell}\cap (\Lambda(6\ell)\setminus\Lambda(3\ell))|\right\rangle\right)&=\E \left( D_{\Lambda(6\ell)\setminus\Lambda(3\ell),\Lambda(6\ell)\setminus\Lambda(\ell-1)}\right)\\
    &\le C_2\cdot\ell\cdot\sum_{r=0}^{3\ell-1} m(r)\le \frac{C_3}{\gamma}\cdot m(\ell)\cdot \ell^{2+\gamma}.
  \end{split}
  \end{equation}
  for absolute constants $C_2, C_3$. Thus~\eqref{eq:third error bound} follows from Markov's inequality.
\end{proof}
Combining~\eqref{eq:surface tension bound2},~\eqref{eq:Cauchy-Schwartz} and the preceding lemma we arrive at
\begin{equation}\label{eq:surface tension bound3}
  \mathbb{E}\left(\mathcal{T}_{\Lambda(\ell) ,\Lambda(6\ell)}\right)\le 16J\left(M + \frac{1}{\ell}\sqrt{\frac{C}{\gamma}\cdot m(\ell)^2\cdot \ell^{4+2\gamma}\cdot\left(C e^{-c' \sqrt{\ell}} + \frac{C}{\gamma}\cdot \frac{m(\ell)\cdot \ell^{1+\gamma-\alpha}}{M}\right)}\right)\,.
\end{equation}
Finally choosing $M:=\frac{1}{\gamma^{2/3}}\cdot m(\ell) \cdot  \ell^{1 -\frac{1}{3}\alpha + \gamma}$ finishes the proof of the proposition.

\subsection{Proof of Proposition~\ref{prop:fast power law}} \mbox{ } \\[-1ex]

Let $\alpha = \alpha(J/\eps)$ and $\ell_1 = \ell_1(J/\eps)$ take their values from Theorem~\ref{thm:Fractality}.
Fix a constant $0<c_0<1$ throughout the proof. Suppose $N$ is an integer satisfying
\begin{equation}\label{eq:contradiction assumption}
  m(N)>\frac{c_0}{N+1}.
\end{equation}
Applying Lemma~\ref{lem:comp_decay} with $(m(j))$ in the role of $(p(j))$, $k=N$ and $\gamma=\frac{1}{12}\alpha$ we obtain the existence of an integer $L\ge 0$ satisfying
\begin{equation}\label{eq:L and N}
  N\ge L\ge (N+1)\cdot m(N)^{1/(1+\frac{1}{12}\alpha)}-1\ge (N+1)\left(\frac{c_0}{N+1}\right)^{1/(1+\frac{1}{12}\alpha)}-1.
\end{equation}
for which
\begin{equation}\label{eq:comp decay end of proof}
  m(L) \leq m(j) \leq m(L) \left(\frac{L+1}{j+1}\right)^{1+ \frac{\alpha}{12}} \quad 0 \leq j \leq L.
\end{equation}
Set $\ell = \lfloor L/5\rfloor$. Assuming that $L>\ell_1$ we may Apply Proposition~\ref{prop:GoodSTUpperBound} to obtain
\begin{equation}
  \mathbb{E}[\mathcal{T}_{\Lambda(\ell) ,\Lambda(6\ell)}] \leq C_1\alpha^{-\frac{2}{3}}\cdot J \cdot m(\ell) \cdot  \ell^{1 - \frac{1}{4}\alpha}
\end{equation}
where we let $C_j$ stand for absolute constants. This bound is used in conjunction with Proposition~\ref{prop:AntiConcentration} to deduce that
\begin{equation}
  \P\left[\frac{D_{\Lambda(\ell) ,\Lambda(6\ell)}}{\E[D_{\Lambda(\ell) ,\Lambda(6\ell)}]} < 1/2 \right] \geq \chi\left(\frac{1} {2 \eps }\cdot  \frac{ C_1\alpha^{-\frac{2}{3}}\cdot J \cdot m(\ell) \cdot  \ell^{1 - \frac{1}{4}\alpha}}{\sqrt{|\Lambda(\ell)|}} \cdot \frac{|\Lambda(\ell)|}{\E[D_{\Lambda(\ell) ,\Lambda(6\ell)}]} \right).
\end{equation}
Using also the monotonicity of disagreement percolation and~\eqref{eq:comp decay end of proof} we find that
\begin{equation}\label{eq:anti-concentration end of proof}
  \P\left(\frac{D_{\Lambda(\ell) ,\Lambda(6\ell)}}{\E(D_{\Lambda(\ell) ,\Lambda(6\ell)})} < 1/2 \right) \geq \chi\left(\frac{C_2}{\alpha^{2/3}}\cdot \frac{J} {\eps }\cdot\ell^{-\frac{1}{4}\alpha} \right),
\end{equation}
noting that the right-hand side is very close to $1$ for $\ell$ large.

In contrast, Proposition~\ref{prop:var_bound} shows that
\begin{equation}
  \E\big((D_{\Lambda(\ell),\Lambda(6\ell)})^2\big)\ \le\ \frac{C_3}{\alpha}\cdot \ell^{\frac{1}{6}\alpha} \cdot \big(\E \left( D_{\Lambda_1,\Lambda_2}\right) \big)^{2}\,.
\end{equation}
This implies, via the one-sided Chebyshev inequality, that
\begin{equation}\label{eq:concentration end of proof}
  \P\left(\frac{D_{\Lambda(\ell) ,\Lambda(6\ell)}}{\E(D_{\Lambda(\ell) ,\Lambda(6\ell)})} < 1/2 \right) < \frac{1}{1+\frac{\alpha}{C_4}\ell^{-\frac{1}{6}\alpha}}.
\end{equation}
The bounds~\eqref{eq:anti-concentration end of proof} and~\eqref{eq:concentration end of proof} are in contradiction for $\ell$ sufficiently large. The contradiction shows that our initial assumption~\eqref{eq:contradiction assumption} is false, thus concluding the proof of the proposition.

Since $\ell = \lfloor L/5\rfloor$, we may track the dependence through~\eqref{eq:L and N} and deduce that
\begin{equation}\label{eq:quantitative dependence}
  N>\frac{1}{c_0}\max\left\{2,\frac{J}{\eps}, \frac{1}{\alpha}, \ell_1\right\}^{\frac{C_5}{\alpha^2}}  \quad \text{implies} \quad m(N) \leq \frac{c_0}{N+1}.
\end{equation}
\section{Open Questions}
The question discussed here remains open for systems of continuous spin variables, e.g. two-component spin models with rotation-invariant interaction.  It is of particular interest with regards to random fields with rotationally-invariant distribution, in dimensions $2 \leq d \leq 4$. \\

A remaining question for the random-field Ising model is the quantitative dependence of the \emph{correlation length} $\xi$ on the disorder strength. There are actually a number of relevant length scales.  One of these, which makes sense even in the absence of exponential decay, is the distance  $\xi_0(\eps/J)$ at which the order parameter drops below some fixed threshold.   Another, which is finite only in the presence of exponential decay, is the inverse of the rate of that decay. The preceding works~\cite{C17, AP18} provide the upper bound
\begin{equation}\label{eq:correlation length}
  \xi_0\le e^{e^{O\left(\left(J/\eps\right)^2\right)}}
\end{equation}
at weak disorder.   A heuristic upper bound of the same general form may be understood from the analysis of Imry and Ma~\cite{IM75} or from the Mandelbrot percolation picture of~\cite{AP18}.   As stated, Theorem~\ref{thm:exponential_bound} does not provide quantitative bounds on the correlation length; however, its proof implies such bounds in terms of the quantities arising from the tortuosity result~\cite{AB99} (see Section~\ref{sec:exp decay}).   It is of interest to determine the general rate of growth of $\xi$. As an alternative to the form appearing in~\eqref{eq:correlation length}, the behavior
\be
\xi \simeq \exp\!\left(O\left(J/\eps)^2\right)\right)
\ee was discussed in~\cite{BK88}.

\section*{Acknowledgements}
We gratefully acknowledge the following support.  The work of MA was supported in parts by the NSF grant DMS-1613296, and the Weston Visiting Professorship at the Weizmann Institute of Science.
The work of MH and RP was supported in part by Israel Science Foundation grant 861/15 and the European Research Council starting grant 678520 (LocalOrder).  MH was supported by the Zuckerman Postdoctoral Fellowship.
We  thank the Faculty of Mathematics and Computer Science and the Faculty of Physics at WIS for the hospitality enjoyed there during work on this project.

\end{document}